\documentclass[journal,twocolumn]{IEEEtran}

\usepackage[cmex10]{amsmath}
\usepackage{amssymb}
\usepackage{cite}
\usepackage{url}
\usepackage{graphicx}
\usepackage{array,color}
\usepackage{multirow}
\usepackage{amsmath}
\usepackage{stfloats}
\usepackage{graphicx}
\usepackage{subfigure}
\usepackage{tabularx}
\usepackage{epsfig,color,balance,cite}
\usepackage{setspace}
\usepackage{bm}
\newcommand{\mathcalbf}[1]{\boldsymbol{\mathcal{#1}}}
\usepackage{textcomp}
\usepackage{algorithmic}
\usepackage{algorithm}
\usepackage{comment}
\usepackage{caption}
\usepackage{graphicx}
\usepackage{setspace}
\usepackage{diagbox}
\usepackage{float}
\usepackage{epstopdf}
\usepackage{caption}
\usepackage{subcaption}
\usepackage{booktabs}
\usepackage{array}
\usepackage{makecell}
\newcolumntype{Y}{>{\centering\arraybackslash}X}
\UseRawInputEncoding
\usepackage[cmex10]{amsmath}
\usepackage{amssymb}
\usepackage{cite}
\usepackage{graphicx}
\usepackage{array,color}
\usepackage{multirow}
\usepackage{amsmath}
\usepackage{stfloats}
\usepackage{graphicx}
\usepackage{subfigure}
\usepackage{tabularx}
\usepackage{epsfig,epsf,color,balance,cite}
\usepackage{setspace}
\usepackage{bm}
\usepackage{textcomp}
\usepackage{algorithmic}
\usepackage{algorithm}
\usepackage{caption} 
\usepackage{comment}
\usepackage{subfigure}

\usepackage{caption}
\usepackage{graphicx}
\usepackage{setspace}
\usepackage{diagbox}

\usepackage{multirow}
\usepackage{graphicx}
\allowdisplaybreaks[4]
\makeatletter

\newcommand{\Rmnum}[1]{\expandafter\@slowromancap\romannumeral #1@}
\makeatother
\usepackage{booktabs}
\usepackage{amsmath}
\newtheorem{Theorem}{Theorem }

\begin{document}
	
	\title{Tensor Train Decomposition-based Channel Estimation for MIMO-AFDM Systems with Fractional Delay and Doppler}
	
	\author{Ruizhe Wang, Cunhua Pan, Hong Ren, Haisu Wu, and Jiangzhou Wang, \IEEEmembership{Fellow, IEEE}
		\thanks{%R. Wang, C. Pan, H. Ren,  H. Wu and J. Wang 
		The authors are with National Mobile Communications Research Laboratory, Southeast University, Nanjing, China. (e-mail: {rzw, cpan, hren,  wuhaisu, j.z.wang}@seu.edu.cn).

		\emph{Corresponding author: Cunhua Pan and Hong Ren}}
	}

	\maketitle

	%=\thanks{This work was supported in part by the National Key Research and Development Project under Grant 2019YFE0123600, National Natural Science Foundation of China (62101128), Basic Research Project of Jiangsu Provincial Department of Science and Technology (BK20210205) and High Level Personal Project of Jiangsu Province (JSSCBS20210105)}}

\begin{abstract}
Affine Frequency Division Multiplexing (AFDM) has emerged as a promising chirp-based multicarrier technology for high-speed communication systems. 
By applying discrete affine Fourier transform (DAFT), AFDM has the potential to achieve significant improvements in spectral efficiency (SE) while maintaining a bit error rate (BER) performance comparable to that of OTFS. 
To fully exploit the diversity gain offered by AFDM, accurate channel estimation is essential.
However, existing studies have mainly focused on the integer-delay-tap scenario and single-symbol pilot-based estimation. 
Since delay taps in practice are generally fractional, approximating them as integers not only degrades delay estimation accuracy but also severely affects Doppler frequency estimation.
To address this problem, in this paper, we investigate channel estimation for multiple-input multiple-output (MIMO)-AFDM systems.
A time-affine frequency (T-AF) domain pilot structure is proposed to exploit time-domain phase variations.
By leveraging the rotational invariance property in the spatial and temporal domains, a channel estimation algorithm based on Vandermonde-structured tensor-train (TT) decomposition is developed. 
The proposed algorithm demonstrates superior computational efficiency compared with state-of-the-art parameter estimation methods.
Moreover, diverging from current studies, we derive the global Ziv-Zakai bound (ZZB) as an alternative parameter estimation error lower bound to the Cram\'{e}r-Rao bound (CRB). 
Numerical results show that the derived ZZB provides tighter global performance characterization and successfully captures the threshold phenomenon in mean square error (MSE) performance in the low-SNR regime.
Furthermore, the proposed algorithm achieves superior communication performance relative to the existing schemes, while offering a computational speedup, reducing the execution time by an order of magnitude compared to the state-of-the-art iterative algorithms.
\end{abstract}

\begin{IEEEkeywords}
	AFDM, tensor decomposition, channel estimation, embedded pilot, fractional delay and Doppler.
\end{IEEEkeywords}

\IEEEpeerreviewmaketitle
%
%\newpage
\section{Introduction}
Ubiquitous connectivity is envisioned as one of the fundamental usage scenarios of future sixth-generation (6G) communication systems, with the objective of providing global connectivity and high-quality communication service. 
{A critical challenge in achieving this vision lies in guaranteeing reliable communication in high-mobility scenarios, including  vehicle-to-everything (V2X), uncrewed aerial vehicles (UAVs), and high-speed trains. }
These applications generally involve extremely high relative velocities, which induce severe Doppler spreads. 
The existing multicarrier technologies, such as orthogonal frequency division multiplexing (OFDM), suffer from significant inter-carrier interference (ICI) in these scenarios, resulting in severe degradation of communication quality \cite{ofdmdopplerspread,Huilingzhu1,Huilingzhu2}.

To fundamentally overcome the performance degradation of OFDM in high-mobility scenarios, researchers have turned to novel modulation waveforms. 
One of the promising approaches is the delay-Doppler (DD) domain modulation, with orthogonal time frequency space (OTFS) being a representative technique \cite{otfs,otfs2}. 
OTFS is a two-dimensional (2D) modulation scheme that transforms the multipath doubly selective channel into sparse and approximately time-invariant impulse responses in the delay-Doppler (DD) domain, thereby improving communication reliability.
However, OTFS also suffers from notable drawbacks.
Firstly, OTFS requires block-wise modulation over multiple symbols, which introduces inherent latency into the system \cite{otfsdrawbacks}.
Secondly, due to its 2D modulation structure, OTFS symbols exhibit intrinsic 2D coupling, preventing symbol-by-symbol processing as in OFDM. 
Instead, the receiver must perform joint detection over 2D signal blocks, which significantly increases detection complexity \cite{OTFSSurvey}.

Chirp signal-based waveform design has also attracted considerable attention. 
Recently, a novel waveform termed affine frequency division multiplexing (AFDM) has been proposed for next-generation wireless communications to mitigate performance degradation in high-mobility scenarios \cite{AFDM}.
AFDM employs the inverse discrete affine Fourier transform (IDAFT) to map data symbols onto multiple orthogonal chirp signals. 
{By properly configuring the affine parameters according to the maximum Doppler frequency, AFDM converts the doubly selective channel into a quasi-static and sparse representation in the one-dimensional discrete affine-frequency (AF)-domain, similar in spirit to OTFS.}
It has been shown in \cite{AFDM} that AFDM can achieve optimal multipath diversity.
Results reported in 
%\cite{AFDM,AFDM1,AFDM2,AFDM3} 
\cite{AFDM,AFDM1} 
demonstrated that AFDM achieves bit error rate (BER) performance comparable to OTFS while significantly outperforming OFDM. 
Furthermore, different from the 2D modulation of OTFS, which needs to buffer multiple symbols for block processing. AFDM involves no inter-symbol coupling. 
Consequently, it is  compatible with the current 5G new radio (NR) frame structure \cite{afdmisac2}.

Nevertheless, the results reported in 
%\cite{AFDM,AFDM1,AFDM2,AFDM3} 
\cite{AFDM,AFDM1} 
rely on the ideal assumption of perfect channel state information (CSI).
In practical deployments, however, CSI is generally imperfect, which necessitates performance evaluation under practical channel estimation.
In \cite{afdmisac1}, a linear minimum mean square error (LMMSE) estimator was derived for single-input-single-output (SISO)-AFDM systems. 
In \cite{diagre}, an embedded pilot-aided diagonal reconstructability algorithm was proposed, which detects partial channel responses and reconstructs the entire AFDM channel by leveraging its diagonal reconstruction property. 
Nevertheless, the above methods did not explicitly estimate the underlying physical channel parameters, including the delay, Doppler and channel gain.
In \cite{AFDM4}, the authors proposed a MIMO-AFDM channel estimation algorithm based on the embedded pilot structure, where the delay and Doppler of each propagation path were assumed to lie on the discrete grid.
The authors of \cite{AFDM6} proposed a superimposed pilot structure and an orthogonal matching pursuit (OMP)-based channel estimation algorithm.
To address fractional Doppler estimation, \cite{AFDM5} developed a sparse Bayesian learning (SBL)-based off-grid channel estimation algorithm, where hyperparameters are updated via the expectation maximization (EM) algorithm, resulting in high computational complexity.
Furthermore, \cite{AFDM7} proposed a grid-evolution SBL algorithm for integer-delay and fractional-Doppler estimation using the first-order Taylor expansion. {However, the existing literature did not account for fractional delays.
It is noteworthy that, within the AFDM input-output framework, the discrete affine-frequency-domain peak of a propagation path is determined by a linear combination of its delay and Doppler frequency. With a sufficiently large number of subcarriers, neglecting the fractional component of delay introduces negligible delay estimation error yet leads to considerable Doppler frequency estimation deviation. While such deviation can be compensated via simple phase correction in single-path channels, it becomes destructive in multipath scenarios, as distinct Doppler frequencies across paths produce varying phase errors that cannot be jointly corrected with a single global factor \cite{phasecorrect}. Furthermore, the integer-delay approximation introduces a subcarrier-dependent phase mismatch, resulting in severe channel-reconstruction errors.}

{Motivated by the aforementioned challenges, we investigate pilot
design and channel estimation for MIMO-AFDM systems accounting for
fractional delay and Doppler. The main contributions of this paper are
summarized as follows:
\begin{itemize}
	\item 
	\textit{\textbf{Fractional delay-Doppler modeling and pilot
    design:}}
    We derive the input-output relationship of MIMO-AFDM systems
    under fractional delay and Doppler. Based on the resulting
    delay-Doppler coupling in the AF-domain, we design a
    time-extended embedded pilot structure that captures the temporal
    phase variation of each propagation path. This additional temporal
    information enables the Doppler frequency to be estimated
    independently of the coupled AF-domain location, thereby
    facilitating fractional-delay recovery.
	\item 
	\textit{\textbf{Low-complexity structured channel estimation:}}
    We model the received pilot as a 3D tensor and propose a sequential TT decomposition leveraging the Vandermonde structures of the spatial and temporal factors. The sequential reconstruction
    preserves the association among the recovered AoA, Doppler, and
    AF-domain factors while avoiding alternating iterative
    optimization. Furthermore, since the fractional-delay AF-domain
    factor is a partially observed Dirichlet response, we develop a
    Nyquist-interpolated linear interpolated discrete time Fourier transform (LIpDTFT) refinement to estimate its continuous
    location without imposing an integer-delay constraint. The
    proposed implementation achieves a lower computational burden than
    the considered iterative estimation methods.
	\item 
	\textit{\textbf{Global performance-bound analysis:}}
    We derive the Ziv-Zakai bound (ZZB) for the channel-parameter
    estimation problem in the considered MIMO-AFDM system. Unlike the
    local Cram\'{e}r-Rao bound (CRB), the ZZB provides a global
    characterization of the estimation error and captures the
    threshold behavior that occurs in the low-SNR regime.
	\item 
	\textit{\textbf{Performance evaluation:}}
    Numerical results validate the proposed parameter-estimation
    framework and demonstrate its communication performance in terms
    of normalized mean squared error (NMSE), bit error rate (BER), and
    spectral efficiency (SE). The results also demonstrate the
    necessity of fractional-delay modeling: even with oracle
    parameters and refitted path gains, the mismatched integer-delay
    model exhibits a pronounced BER floor. In addition, the proposed
    non-iterative implementation provides a substantial reduction in
    CPU runtime compared with the considered iterative benchmarks.
\end{itemize}}

The rest of this paper is organized as follows. Section \ref{SystemModel} describes the MIMO-AFDM system model. 
In Section \ref{PSDCE}, we propose a time-affine frequency (T-AF) domain embedded pilot structure and channel estimation algorithm. 
The ZZB analysis is provided in Section \ref{ZZBAnalysis}. 
Section \ref{SimulationResults} provides simulation results. 
Finally, conclusions are drawn in Section \ref{Conclusion}.
\vspace{-0.5cm}

\section{MIMO-AFDM System Model}\label{SystemModel}

\begin{figure*}[ht]
	\centering
	\begin{minipage}[t]{1.0\linewidth}
		\centering
		\includegraphics[width=1.0\linewidth]{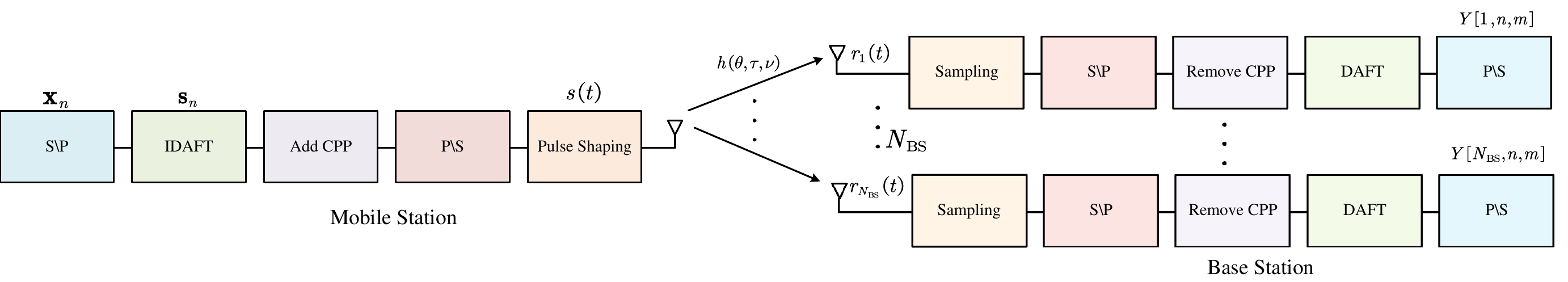}
		\caption{An example of MIMO-AFDM systems.}
		\label{systemmodel}
	\end{minipage}%
	\vspace{-0.5cm}
	\hfill
\end{figure*}

Consider a MIMO-AFDM system shown in Fig.~\ref{systemmodel}. 
The mobile station (MS) is equipped with a single antenna, while the base station (BS) employs a uniform linear array (ULA) consisting of $N_{\mathrm{BS}}$ antennas. The total number of subcarriers is $M$. The sampling interval is given by $T_s=\frac{1}{M\Delta f}$, where $\Delta f$ represents the subcarrier spacing. The duration of the chirp-periodic prefix (CPP) is $T_{\mathrm{CPP}}=M_{\mathrm{CPP}}T_s$. Let $T_{\mathrm{sym}}$ and $T$ denote the symbol durations with and without the CPP, respectively. Consequently, the symbol duration without CPP is $T=\frac{1}{\Delta f}$ and $T=\frac{M}{M+M_{\mathrm{CPP}}}T_{\mathrm{sym}}$. Let $x_{n,m}$ denote the symbol transmitted on the $m$-th subcarrier of the $n$-th AFDM symbol. Denote $\mathbf{x}_n=\left[ x_{n,0},x_{n,1},...,x_{n,M-1} \right] ^{\mathrm{T}}$ and the corresponding modulated signal as $\mathbf{s}_n=\left[ s_{n,0},s_{n,1},...,s_{n,M-1} \right] ^{\mathrm{T}}$. According to \cite{AFDM}, the modulated signal of the $n$-th AFDM symbol is given by
\begin{equation}
	\mathbf{s}_n=\mathbf{\Lambda }_{c1}^{\mathrm{H}}\mathbf{F}_{M}^{\mathrm{H}}\mathbf{\Lambda }_{c2}^{\mathrm{H}}\mathbf{x}_n,
\end{equation}
where $\mathbf{\Lambda }_c=\mathrm{Diag}\left( e^{-j2\pi cm^2},m=0,1,...,M-1 \right) $, and $\mathbf{F}_M$ is the $M\times M$ DFT matrix.
The analog time domain symbols with rectangular waveform shaping can be then written as
\begin{equation}\label{st}
s\left( t \right) =\sum_{n=1}^{N}\sum_{m=1}^{M}x_{n,m}\phi _m\left( t-nT_{\mathrm{sym}} \right) g_{\mathrm{tx}}\left( t-nT_{\mathrm{sym}} \right), 
\end{equation}
where the AFDM chirp basis function $\phi _m$ for the $m$-th subcarrier is given by
\begin{equation}\label{phi_m}
	\begin{aligned}
		\phi _m\left( t \right) &=e^{j\left( \frac{d}{2b}\left( t-T_{\mathrm{CPP}} \right) ^2+2\pi \left( m\Delta f \right) \left( t-T_{\mathrm{CPP}} \right) +\frac{a}{2b}\left( m\Delta f \right) ^2 \right)},
		\\ 
		&\qquad\qquad\qquad\qquad\qquad\qquad 0\le t\le T_{\mathrm{sym}},
	\end{aligned}
\end{equation}
where $\left( a,b,c,d \right) $ are the parameters of affine transform.
The rectangular waveform function\footnote{ Rectangular pulses are adopted for analytical tractability and consistency with the conventional discrete-time AFDM model. Their abrupt edges may cause high out-of-band emission (OOBE). For known shaping pulses accommodated by the CPP/guard interval, the spatial and temporal factors remain unchanged, while the AF-domain response becomes pulse-dependent. Thus, the proposed TT-based AoA/Doppler estimation remains applicable, whereas delay and gain estimation require pulse-dependent recalibration.} $g_{\mathrm{tx}}\left( t \right) $ is given by
\begin{equation}\label{rectangularwave}
	g_{\mathrm{tx}}\left( t \right) =\begin{cases}
	\frac{1}{\sqrt{T}}, 0\le t\le T_{\mathrm{sym}}\\
	0, \mathrm{otherwise}.\\
\end{cases}
\end{equation}
Note that the CPP has already been included in (\ref{phi_m}). In AFDM, we have $c_1=\frac{d}{4\pi b}T_{s}^{2}$ and $c_2=\frac{a}{4\pi b}\Delta f^2$. Generally, AFDM can achieve full diversity by setting $c_1=\frac{2\left( \alpha _{\max}+k_{\nu} \right) +1}{2M}$ and $c_2$ to a value much smaller than $\frac{1}{2M}$ \cite{AFDM}, where $\alpha _{\max}=\nu _{\max}T$ and $\nu _{\max}$ denotes the maximum Doppler frequency, and $k_{\nu}$ denotes the spacing factor. 
The received signal $\mathbf{r}\left( t \right) $ through the multipath time-varying channel at the BS is given by
\begin{equation}
	\begin{aligned}
		\mathbf{r}\left( t \right) &=\iiint{}\mathbf{a}_{N_{\mathrm{BS}}}\left( \theta \right) h\left( \theta ,\tau ,\nu \right) e^{j2\pi \nu \left( t-\tau \right)}
		\\
		&\quad\times s\left( t-\tau \right) \mathrm{d}\tau \mathrm{d}\nu \mathrm{d}\theta +\mathbf{n}\left( t \right),
	\end{aligned}
\end{equation}
where the space-delay-Doppler channel response $h\left( \theta ,\tau ,\nu \right) $ and the array steering vector $\mathbf{a}_{N_{\mathrm{BS}}}\left( \theta \right) $ are defined as
\begin{equation}
	h\left( \theta ,\tau ,\nu \right) =\sum_{i=1}^{P}{\alpha _i}{\bm \delta} \left( \theta -\theta _i \right) {\bm \delta} \left( \tau -\tau _i \right) {\bm \delta} \left( \nu -\nu _i \right),
\end{equation}
where $P$ denotes the number of propagation paths
and $
	\mathbf{a}_X\left( \theta \right) =\left[ 1, e^{j2\pi \frac{d}{\lambda}\cos \theta},..., e^{j2\pi \frac{d}{\lambda}\left( X-1 \right) \cos \theta} \right] ^{\mathrm{T}}$. $\mathbf{n}\left( t \right) $ denotes the additive white Gaussian noise vector and each element
	 follows an i.i.d. complex Gaussian distribution $\mathcal{C} \mathcal{N} \left( 0,\sigma _{\mathrm{n}}^{2} \right) $. The $\alpha_i$, $\theta_i $, $\tau_i $ and $\nu_i$  represent the complex channel gain, the arrival of angle (AoA), delay, and the Doppler frequency of the $i$-th propagation path, respectively. The delay and Doppler taps for the $i$-th path can be denoted as
\begin{equation}\label{frdelaydoppler}
	\tau _i=\frac{l_i+\iota _i}{M\Delta f}, \nu _i=\frac{k_i+\kappa _i}{T},
\end{equation}
where $l_i$ and $k_i$ are the integer indices of the delay tap and Doppler tap for the $i$-th propagation path, and $\iota_i \in (-\frac{1}{2}, \frac{1}{2}]$ and $\kappa_i \in (-\frac{1}{2}, \frac{1}{2}]$ are the fractional delay and Doppler for the $i$-th propagation path, respectively. After sampling, the received signals are demodulated via the discrete affine Fourier transform (DAFT) \cite{AFDM}. We have the following theorem to characterize the input-output relationship in the spatial-time-affine frequency (STAF) domain.

\begin{Theorem}\label{theorem1}
	The input-output relationship of the MIMO-AFDM system can be expressed as
	\begin{equation}\label{iorelationship}
		\begin{aligned}
Y\left[ n_{\mathrm{BS}},n,m \right] &=\sum_{m^{\prime}=0}^{M-1}{}x_{n,m^{\prime}}\sum_{i=1}^{P}{}\tilde{\alpha}_i\left[ \mathbf{a}_{N_{\mathrm{BS}}}\left( \theta \right) \right] _{n_{\mathrm{BS}}}
\\
&\quad\times e^{j2\pi \nu _inT_{\mathrm{sym}}}H_i\left[ m,m^{\prime} \right] +N\left[ n_{\mathrm{BS}},n,m \right],
		\end{aligned}
	\end{equation}
where $N\left[ n_{\mathrm{BS}},n,m \right] $ is the additive noise and
\begin{align}
	H_i\left[ m,m^{\prime} \right] &=e^{j\frac{2\pi}{M}\left( Mc_1\left( l_i+\iota _i \right) ^2-m^{\prime}\left( l_i+\iota _i \right) +Mc_2\left( {m^{\prime}}^2-m^2 \right) \right)}\nonumber
\\
&\quad \times\frac{e^{-j2\pi \left( m-m^{\prime}-\mathrm{loc}_i \right)}-1}{Me^{-j\frac{2\pi}{M}\left( m-m^{\prime}-\mathrm{loc}_i \right)}-M}\nonumber
\\
\tilde{\alpha}_i&=\alpha _ie^{j2\pi \nu _i\left( \frac{M_{\mathrm{CP}}}{M}T-\tau _i \right)},\nonumber\label{Hmm}
\\
\mathrm{loc}_i&=\nu _iT-2Mc_1\left( l_i+\iota _i \right).
\end{align}
respectively.

\end{Theorem}

\textit{Proof:} The proof is provided in Appendix \ref{proofoftheorem1}. \hfill $\blacksquare$

From the input-output relationship (\ref{iorelationship}), the following three conclusions can be drawn. 

\indent 1. The effective channel response on the $m^{\prime}$-th subcarrier of the $n$-th transmitted AFDM symbol is $H_{\mathrm{eff},m^{\prime}}\left[ n_{\mathrm{BS}},n,m \right] =\sum_{i=1}^P{}\tilde{\alpha}_i\left[ \mathbf{a}_{N_{\mathrm{BS}}}\left( \theta \right) \right] _{n_{\mathrm{BS}}}e^{j2\pi \nu _inT_{\mathrm{sym}}}H_i\left[ m,m^{\prime} \right] $, where a phase rotation $e^{j 2\pi \nu_i T_{\mathrm{sym}}}$ is induced by the Doppler frequency on the $i$-th propagation path. 
The channel estimation for the MIMO-AFDM system is the estimation of $\left[ \mathcalbf{H} ^{\mathrm{STAF}} \right] _{n_{\mathrm{BS}},n,m}=H_{\mathrm{eff},m^{\prime}}\left[ n_{\mathrm{BS}},n,m \right] $.

\indent 2. The AFDM effective channel exhibits diagonal reconstructability \cite{diagre}. Specifically, the complete AF-domain channel can be reconstructed once the channel parameters are estimated from the support set associated with a certain symbol on the $m^{\prime}$ subcarrier. 

\indent 3. When the integer delay constraint is relaxed, the number of unknown parameters to be estimated increases. Consequently, the maximum likelihood estimation algorithm for the delay and Doppler using a single AFDM symbol becomes imprecise \cite{AFDM,ISACAFDM1,ISACAFDM2}.

The above analysis suggests that single-symbol estimation is insufficient for resolving fractional delay and Doppler shifts accurately. 
The assumption of integer delay neglects the fractional part $\left\{ \iota _i \right\} $, which not only results in a coarse delay estimation but also leads to severe deviations in the Doppler frequency estimation $\left| \Delta \nu _i \right|=\left| 2Mc_1\iota _i/T \right|$. Typically, $c_1=\frac{2\left( \alpha _{\max}+k_{\nu} \right) +1}{2M}$, and this error is substantial and cannot be neglected.
Furthermore, from a positioning perspective, this assumption can result in decimeter-level positioning errors.
For example, in the context of 5G NR numerologies, the subcarrier spacing is generally set to $2^{\mu}\cdot 15$ kHz, where $\mu =0,1,...,6$. 
With 1024 subcarriers and a 30 kHz spacing, the system achieves a fine delay resolution of 32.55 ns, equivalent to a spatial resolution of 9.76 m. 
Thus, for both communication and positioning tasks, it is imperative to incorporate fractional delay and Doppler into the system model.

\section{Pilot Structure Design And Channel Estimation Algorithm}\label{PSDCE}
In this section, we first introduce the proposed MIMO-AFDM embedded pilot structure and formulate the channel estimation problem. 
Then, by exploiting the multidimensional structure and inherent low-rank properties of the MIMO-AFDM signals, we develop a computationally efficient channel estimation framework based on Vandermonde-structured tensor-train (TT) decomposition tailored to the proposed pilot structure.
\vspace{-0.5cm}

\subsection{Proposed T-AF Embedded Pilot Structure}
\begin{figure}[t]
	\centering
	\begin{minipage}[t]{1.0\linewidth}		\centering
		\includegraphics[width=1.0\linewidth]{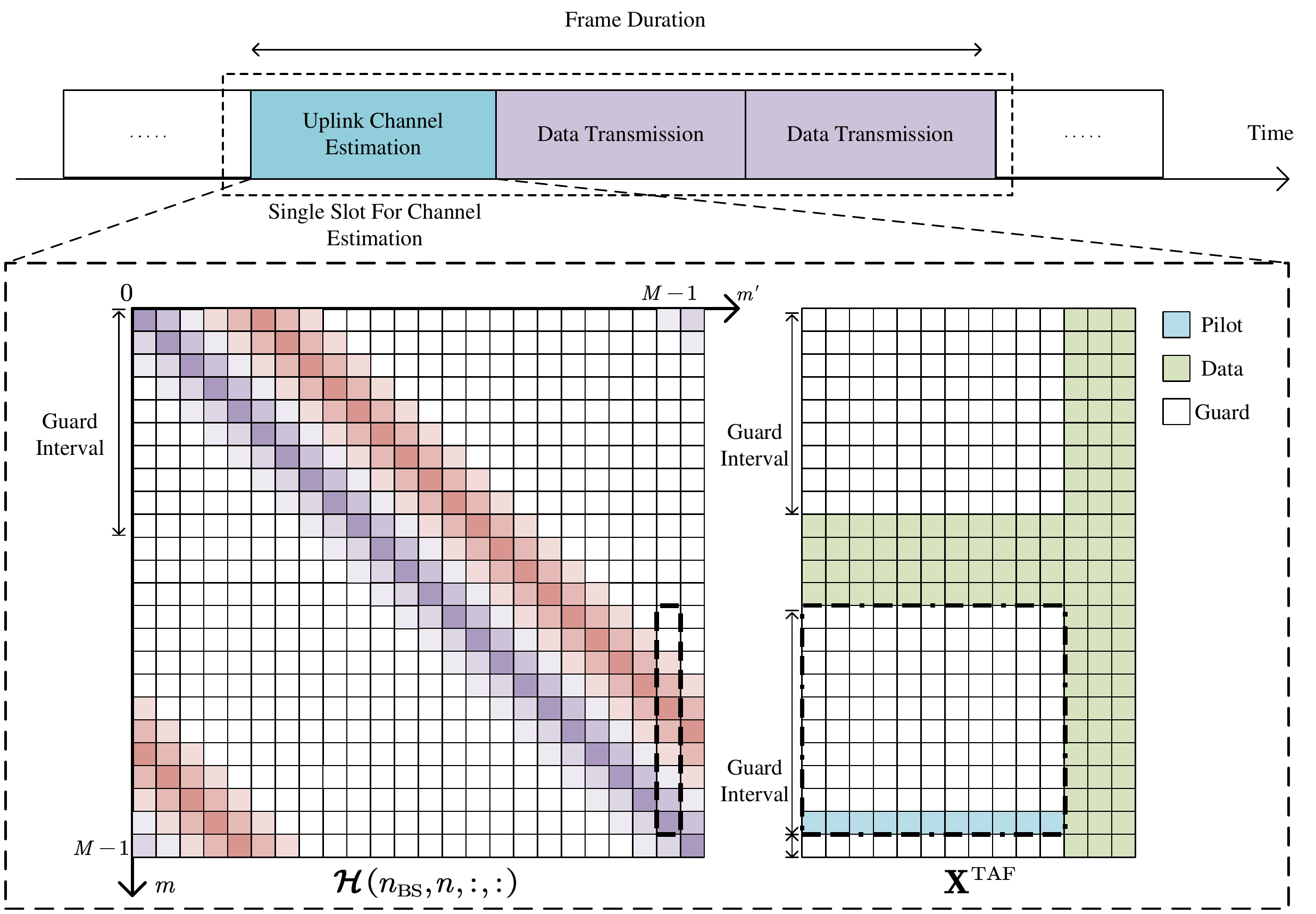}
		\caption{Proposed TAF-domain embedded pilot structure.}
		\label{mimoafdmpilot}
	\end{minipage}%
	\vspace{-0.5cm}
	\hfill
\end{figure}

Fig.~\ref{mimoafdmpilot} shows the proposed time-affine-frequency (TAF)-domain embedded pilot structure. 
In the first slot of each frame, pilot symbols are inserted into the first several AFDM symbols, and guard intervals are placed between the pilot and data symbols to prevent mutual interference. The remaining symbols are allocated for data transmission.
The left half of Fig.~\ref{mimoafdmpilot} presents an example of the equivalent AFDM channel with two propagation paths. 
It can be observed that, for the $m'$-th column of the equivalent affine-frequency domain channel $\mathcalbf{H}(n_{\mathrm{BS}},n,:,m')$, the amplitude pattern in the dashed region exhibits a cyclic shift relative to that of the adjacent columns. This phenomenon corresponds to the so-called diagonal reconstructability property reported in \cite{diagre}.

However, according to Theorem \ref{theorem1}, when both delay and Doppler are fractional, the equivalent intercept corresponding to the slanted line of the channel response for the $i$-th propagation path in $\mathcalbf{H} \left( n_{\mathrm{BS}},n,:, m^{\prime} \right) $ is given by (\ref{Hmm}).
As previously discussed, adhering to the integer delay assumption under these conditions results in coarse delay estimation and incurs severe errors in Doppler estimation.
Consequently, different from \cite{AFDM}, the pilot structure shown in Fig.~\ref{mimoafdmpilot} incorporates embedded pilots along the time dimension to extract the time domain phase-varying feature, which is marked by the dashed-dotted line. 
By exploiting such phase characteristics, the fractional delays and Doppler can be simultaneously estimated without loss of accuracy. 

Let $m_{\mathrm{pilot}}$ and $N$ denote the position of the embedded pilot symbol in each AFDM symbols and the number of AFDM symbols that are inserted pilot symbols, respectively. By collecting all the pilot signals received on the  $n_{\mathrm{BS}}$-th antenna, we have
\begin{equation}
	\widehat{\mathbf{Y}}_{n_{\mathrm{BS}}}=\sum_{i=0}^{P-1}{}\tilde{\alpha}_ie^{j2\pi \frac{d}{\lambda}\left( n_{\mathrm{BS}}-1 \right) \cos \theta _i}\mathbf{b}_i\mathbf{c}_{i}^{\mathrm{T}}+\mathbf{N}_{n_{\mathrm{BS}}},
\end{equation}
where 
\begin{align}
	\mathbf{b}_i\left( n \right) &=e^{j2\pi \nu _inT_{\mathrm{sym}}}
\\
\mathbf{c}_i\left( m \right) &=H_i\left[ m,m_{\mathrm{pilot}} \right] X\left[ m_{\mathrm{pilot}} \right].
\end{align}
By further stacking the signals in the space-TAF (STAF)-domain, we can obtain the following tensor model in the CPD format as
\begin{equation}\label{tensorY}
	\widehat{\mathcalbf{Y}} =\mathcalbf{S} \times _1\mathbf{A}_1\times _2\mathbf{A}_2\times _3\mathbf{A}_3+\mathcalbf{N},
\end{equation}
where the core tensor $\mathcalbf{S}$ and the factor matrices $\mathbf{A}_i, i=1,2,3$ are given by
\begin{align}\label{coreandfactors}
\mathcalbf{S} &= \mathcalbf{D} \left( \left[ \tilde{\alpha}_1,\tilde{\alpha}_2,...,\tilde{\alpha}_{P} \right] ^{\mathrm{T}} \right) \in \mathbb{C} ^{P\times P\times P}\nonumber
\\
\mathbf{A}_1&=\left[ \mathbf{a}_{N_{\mathrm{BS}}}\left( \theta _1 \right) ,\mathbf{a}_{N_{\mathrm{BS}}}\left( \theta _2 \right) ,...,\mathbf{a}_{N_{\mathrm{BS}}}\left( \theta _{P} \right) \right] \in \mathbb{C} ^{N_{\mathrm{BS}}\times P}\nonumber
\\
\mathbf{A}_2&=\left[ \mathbf{b}_1,\mathbf{b}_2,...,\mathbf{b}_{P} \right] \in \mathbb{C} ^{N\times P}\nonumber
\\
\mathbf{A}_3&=\left[ \mathbf{c}_1,\mathbf{c}_2,...,\mathbf{c}_{P} \right] \in \mathbb{C} ^{M_{\mathrm{region}}\times P},
\end{align}
where the entries of $\mathbf{b}_i$ and $\mathbf{c}_i$ are given by $\mathbf{b}_i\left( n \right) =e^{j2\pi \nu _inT_{\mathrm{sym}}}$ and $\mathbf{c}_i\left( m \right) =H_i\left[ m,m_{\mathrm{pilot}} \right] X\left[ m_{\mathrm{pilot}} \right] $, respectively.
From (\ref{coreandfactors}), it can be seen that the three dimensions of the tensor $\widehat{\mathcalbf{Y}}$ contain the angle information $\left\{ \theta _i \right\} $, the Doppler information $\left\{ \nu _i \right\} $, and the intercept $\left\{ \mathrm{loc}_i \right\} $, respectively.
Therefore, MIMO-AFDM channel estimation under fractional delay and Doppler can be transformed into the tensor decomposition and parameter estimation problem.

\subsection{Proposed Tensor Train based Channel Estimation Method}
In this subsection, we focus on the decomposition of $\widehat{\mathcalbf{Y}}$ and channel parameter estimation.

\subsubsection{Factor Matrices Retrieval}

From (\ref{tensorY}) and (\ref{coreandfactors}), it can be seen that the tensor decomposition problem in (\ref{tensorY}) can be solved by canonical polyadic decomposition (CPD). 
Representative methods for solving such a problem include a series of alternating iterative algorithms derived from the core ALS approach, such as the rectified ALS (RecALS) \cite{RecALS}, the SVD-employed ALS \cite{svd-als} and  tensorial minimum mean square error (T-MMSE) \cite{TMMSE}. However, the robustness and convergence of ALS face significant challenges when the tensor rank exceeds 2 \cite{alsconverge}, and the computational efficiency of ALS-based algorithms remains insufficient for high-mobility scenarios.
As a novel tensor decomposition method, TT decomposition \cite{TTD} decomposes a high-dimensional tensor into multiple three-dimensional tensors connected in a chain structure, which offers greater flexibility and higher computational efficiency compared to CPD. 
Moreover, the decomposition error of the TT decomposition can be strictly upper bounded by the sum of squared singular values discarded in the SVD of each dimension, effectively addressing the insufficient convergence and robustness of ALS. 
Without loss of generality, in the following derivations, we use $\mathcalbf{Y} =\mathcalbf{S} \times _1\mathbf{A}_1\times _2\mathbf{A}_2\times _3\mathbf{A}_3$ to denote the noiseless signal model and $\widehat{\mathcalbf{Y}}$ to denote the noisy signal model.
{For notational clarity, the following derivation assumes that the
number of resolvable paths $P$ is known at the
receiver.\footnote{ In practice, $P$ can be estimated before TT-SVD by
applying a standard model-order selection criterion, such as MDL
\cite{1164557}, to the eigenvalues of a sample covariance matrix
constructed from a mode unfolding of $\widehat{\mathcalbf{Y}}$. The
resulting estimate $\widehat P$ is then used to set the TT ranks.}}
The TT decomposition represents tensor $\mathcalbf{Y} $ as a chain-like product of three tensor cores as 
	\begin{align}
\mathcalbf{Y} &=\sum_{i_1=1}^P{}\sum_{i_2=1}^P{}\mathbf{G}_1\left( :,i_1 \right)  \mathcalbf{G} _2\left( i_1,:,i_2 \right) \mathbf{G}_3\left( i_2,: \right) \nonumber
\\
&=\mathbf{G}_1\times _{2}^{1}\mathcalbf{G} _2\times _{3}^{2}\mathbf{G}_3,
	\end{align}
where $\mathbf{G}_1\in \mathbb{C} ^{N_{\mathrm{BS}}\times P}$, $\mathcalbf{G} _2\in \mathbb{C} ^{P\times N\times P}$ and $\mathbf{G}_3\in \mathbb{C} ^{P\times M_{\mathrm{region}}}$ denote the head matrix, the core tensor and the tail matrix, respectively. In this paper, the TT-SVD is applied to reconstruct the TT-cores $\mathbf{G}_1$, $\mathcalbf{G} _2$ and $\mathbf{G}_3$ \cite{vtt}, which can be represented by 
\begin{align}
\mathbf{U}_1\mathbf{V}_1&=\mathbf{A}_1\left( \mathbf{A}_3\odot \mathbf{A}_2 \right) ^{\mathrm{T}}=\mathbf{U}_1\mathbf{M}_1\left( \mathbf{A}_3\odot \mathbf{A}_2 \right) ^{\mathrm{T}}\nonumber
\\
\mathbf{G}_1&=\mathbf{A}_1\mathbf{M}_{1}^{-1}\nonumber
\\
\mathbf{V}_1&=\mathbf{M}_1\left( \mathbf{A}_3\odot \mathbf{A}_2 \right) ^{\mathrm{T}}\label{A1}
\\
\mathbf{U}_2\mathbf{V}_2&=\mathrm{reshape}\left( \mathbf{V}_1;PN,M_{\mathrm{region}} \right) \nonumber
\\
&=\left( \mathbf{A}_2\odot \mathbf{M}_1 \right) \mathbf{A}_{3}^{\mathrm{T}}=\mathbf{U}_2\mathbf{M}_2\mathbf{A}_{3}^{\mathrm{T}}\nonumber
\\
\mathbf{U}_2&=\left( \mathbf{A}_2\odot \mathbf{M}_1 \right) \mathbf{M}_{2}^{-1}, \nonumber
\\
\mathcalbf{G} _2&=\mathcalbf{I} \times _1\mathbf{M}_1\times _2\mathbf{A}_2\times _3\mathbf{M}_{2}^{-\mathrm{T}}\label{A2}
\\
\mathbf{G}_3&=\mathbf{V}_2=\mathbf{M}_2\mathbf{A}_{3}^{\mathrm{T}},\label{A3}
\end{align}
where $\mathbf{M}_1$ and $\mathbf{M}_2$ are the non-singular change-of-bias matrices. The singular value matrix is assumed to be incorporated into the right singular matrix. It is observed from the derivation that the CPD factor matrices are intrinsically embedded within the subspaces of the TT-cores. Specifically, the column spaces of $\mathbf{A}_1$, $\mathbf{A}_2$ and $\mathbf{A}_3$ coincides with the column space of $\mathbf{G}_1$, the subspace spanned by the mode-2 fibers of $\mathcalbf{G} _2$ and the row space of $\mathbf{G}_3$, respectively. 
Given the Vandermonde nature of $\mathbf{A}_1$ and $\mathbf{A}_2$, their defining parameters can be retrieved from the subspaces of $\mathbf{G}_1$ and the subspaces of the mode-2 unfolding of $\mathbf{G}_2$ via ESPRIT. Nevertheless, this approach incurs a column  ambiguity. 
Since the parameters are estimated separately, there is no inherent mechanism to align the path indices, potentially resulting in mismatched pairs of channel parameters for the propagation paths. 

To resolve the ambiguity in parameter pairing, we leverage the intrinsic properties of the Vandermonde matrices to perform a sequential reconstruction of three factor matrices. According to (\ref{A1}), we have
\begin{equation}\label{G1estZ}
	\begin{aligned}
	\underline{\mathbf{G}}_1\mathbf{M}_1&=\underline{\mathbf{A}}_1
\\
\overline{\mathbf{G}}_1\mathbf{M}_1&=\overline{\mathbf{A}}_1=\underline{\mathbf{G}}_1\mathbf{M}_1\mathbf{Z}_1,
	\end{aligned}
\end{equation}
where $\underline{\mathbf{G}}_1$ and $\overline{\mathbf{G}}_1$ represent the matrix ${\mathbf{G}}_1$ with its last row removed and its first row removed, respectively. Then the generators $\left\{ z_{1,r} \right\} $ of Vandermonde matrix $\mathbf{A}_1$ can be estimated via the ESPRIT algorithm
\begin{equation}\label{EVDG1}
\underline{\mathbf{G}}_{1}^{\dagger}\overline{\mathbf{G}}_1=\widehat{\mathbf{M}}_1\mathbf{Z}_1\widehat{\mathbf{M}}_{1}^{-1},
\end{equation}
where $\widehat{z}_{1,r}=\mathbf{Z}_1\left( r,r \right) /\left| \mathbf{Z}_1\left( r,r \right) \right|$. Consequently, the Vandermonde factor matrix $\mathbf{A}_1$ can be reconstructed with permutation ambiguity as $\widehat{\mathbf{A}}_1=\mathbf{A}_1\mathbf{\Pi }$, where $\mathbf{\Pi }$ is a permutation matrix. Next, we use the estimated $\widehat{\mathbf{A}}_1$ to construct a pseudoinverse matrix and then recover $\mathbf{A}_2$. Multiplying $\mathbf{G}_1$ and the mode-2 unfolding of $\mathcalbf{G} _2$, we have 
\begin{equation}\label{G1G2}
	\mathbf{G}_1\left[ \mathcalbf{G} _2 \right] _1=\mathbf{A}_1\left( \mathbf{M}_{2}^{-\mathrm{T}}\odot \mathbf{A}_2 \right) ^{\mathrm{T}}.
\end{equation}
Left-multiplying the term by the pseudoinverse of $\widehat{\mathbf{A}}_1$ yields
\begin{align}\label{pinvAG1G2}
\widehat{\mathbf{A}}_{1}^{\dagger}\mathbf{G}_1\left[ \mathcalbf{G} _2 \right] _1&=\mathbf{\Pi }_{}^{-1}\left( \mathbf{M}_{2}^{-\mathrm{T}}\odot \mathbf{A}_2 \right) ^{\mathrm{T}}\nonumber
\\
&=\left( \mathbf{M}_{2}^{-\mathrm{T}}\mathbf{\Pi }\odot \mathbf{A}_2\mathbf{\Pi } \right) ^{\mathrm{T}}\nonumber
\\
&=\left( \widehat{\mathbf{M}}_{2}^{-\mathrm{T}}\odot \widehat{\mathbf{A}}_2 \right) ^{\mathrm{T}}.
\end{align}
Each row in $\left( \mathbf{M}_{2}^{-\mathrm{T}}\mathbf{\Pi }\odot \mathbf{A}_2\mathbf{\Pi } \right) ^{\mathrm{T}}$ is a vector Kronecker product $\left( \widehat{\mathbf{M}}_{2}^{-\mathrm{T}}\odot \widehat{\mathbf{A}}_2 \right) _{r,:}^{\mathrm{T}}=\left( \widehat{\mathbf{M}}_{2}^{-\mathrm{T}} \right) _{:,r}^{\mathrm{T}}\otimes \left( \widehat{\mathbf{A}}_2 \right) _{:,r}^{\mathrm{T}}$.
By reshaping each row of the Khatri-Rao product $\left( \widehat{\mathbf{M}}_{2}^{-\mathrm{T}} \right) _{:,r}^{\mathrm{T}}\otimes \left( \widehat{\mathbf{A}}_2 \right) _{:,r}^{\mathrm{T}}$ into matrix ${\mathbf{a}}_{2,r}\left( \mathbf{M}_{2}^{-\mathrm{T}} \right) _{:,r}^{\mathrm{T}}$, the factor matrix $\mathbf{A}_2$ and the nonsingular change-of-basis matrix $\mathbf{M}_{2}$ can be recovered by solving the following problem
\begin{align}
		 &\left\{ \widehat{\mathbf{a}}_{2,r},\left( \widehat{\mathbf{M}}_{2}^{-\mathrm{T}} \right) _{:,r}^{\mathrm{T}} \right\} \nonumber
		 \\
		 &\qquad=\mathop {\mathrm{arg}\min} \limits_{\widehat{\mathbf{a}}_{2,r},\left( \widehat{\mathbf{M}}_{2}^{-\mathrm{T}} \right) _{:,r}^{\mathrm{T}}}\left\| {\mathbf{a}}_{2,r}\left( \mathbf{M}_{2}^{-\mathrm{T}} \right) _{:,r}^{\mathrm{T}}-\widehat{\mathbf{a}}_{2,r}\left( \widehat{\mathbf{M}}_{2}^{-\mathrm{T}} \right) _{:,r}^{\mathrm{T}} \right\|\label{a2M2T}. 
\end{align}
The solution can be obtained by performing SVD of $\widehat{\mathbf{a}}_{2,r}\left( \mathbf{M}_{2}^{-\mathrm{T}} \right) _{:,r}^{\mathrm{T}}$. The estimated $\widehat{\mathbf{A}}_2$ and $\widehat{\mathbf{M}}_{2}^{-\mathrm{T}}$ exhibit permutation and scaling ambiguity as 
\begin{align}
	\widehat{\mathbf{A}}_2&=\mathbf{A}_2\mathbf{\Pi \Delta }_{}^{-1}, 
\\
\widehat{\mathbf{M}}_{2}^{-\mathrm{T}}&=\mathbf{M}_{2}^{-\mathrm{T}}\mathbf{\Pi \Delta },
\end{align}
where $\mathbf{\Delta}$ is a diagonal scaling ambiguity matrix. 
Finally, the factor matrix ${\mathbf{A}}_3$ with ambiguity can be obtained as 
\begin{equation}\label{G3}
\widehat{\mathbf{A}}_3=\mathbf{G}_{3}^{\mathrm{T}}\widehat{\mathbf{M}}_{2}^{-\mathrm{T}}=\mathbf{A}_3\mathbf{\Pi \Delta }.
\end{equation}

\subsubsection{Channel Parameters Estimation}
The channel parameters are estimated from the recovered factor matrices. The angles of arrival (AoAs) and Doppler frequencies can  be directly estimated by exploiting the rotational invariance property of the factor matrices.
\begin{align}
	\widehat{\theta}_r&=\mathrm{a}\cos \left( \frac{\lambda}{2\pi d}\angle \left( \mathbf{Z}_1 \right) _{r,r} \right),\label{estAOA}
	\\
	\widehat{\nu}_r&=\frac{1}{2\pi T_{\mathrm{sym}}}\angle \left( \widehat{\mathbf{A}}_{2}^{\dagger}\left( 1:N-1,r \right) \widehat{\mathbf{A}}_2\left( 2:N,r \right) \right),\label{estDoppler}
\end{align}
respectively.

{We then focus on delay estimation. According to (\ref{Hmm}), for
$m^{\prime}=m_{\mathrm{pilot}}$, the AF-domain response of the
$r$-th path exhibits a shifted Dirichlet structure. In the absence
of delay and Doppler, its peak is located at
$m=m^{\prime}=m_{\mathrm{pilot}}$. A nonzero delay and Doppler shift
the continuous peak location to
\begin{equation}
	\begin{aligned}
m_r&=m_{\mathrm{pilot}}+\nu _rT-2Mc_1\left( l_r+\iota _r \right) 
\\
&=\mathrm{round}\left( m_r \right) +\tilde{m}_r,
	\end{aligned}
\end{equation}
with $\tilde{m}_r\in \left[ -\frac{1}{2},\frac{1}{2} \right) $. 
Accordingly, the AF-domain
location parameter is given by
$\operatorname{loc}_r=m_r-m_{\mathrm{pilot}}$. Since $m_r$ generally
does not lie on an integer AF-domain bin, estimating
$\operatorname{loc}_r$ is equivalent to estimating the continuous
location of an off-grid spectral line from samples of its Dirichlet
response. Therefore, conventional peak-bin detection is limited by
the integer-bin resolution and cannot accurately recover the
fractional component $\widetilde{m}_r$.
The LIpDTFT estimator \cite{Belega2024LIpDTFT} provides
super-resolution frequency estimation by interpolating the spectral
samples around an initial peak estimate. However, its conventional
formulation assumes that the complete set of DFT samples is
available, whereas the embedded pilot structure considered here
provides only a partial AF-domain observation through
$\widehat{\mathbf A}_3$. Directly applying the conventional LIpDTFT
would implicitly treat the unobserved AF-domain entries as zeros and
distort the interpolated response. We therefore develop a modified
LIpDTFT estimator that performs Nyquist interpolation directly over
the observed entries of $\widehat{\mathbf A}_3$, enabling continuous
AF-domain location estimation without zero-filling the unobserved
samples.

Let $\mathcal Q=\{1,\ldots,M_{\mathrm{region}}\}$ denote the row-index set of $\widehat{\mathbf A}_3$. The global AF-domain index corresponding to its $q$-th row is $\bar m_q=M-M_{\mathrm{region}}+q-1$. Since the known factor $\exp(-j2\pi c_2\bar m_q^2)$ does not contain any unknown channel parameter, it is first removed as
\begin{align}
	\bar m_q&=M-M_{\mathrm{region}}+q-1,\\
	\check a_{3,r}(q)&=\left[\widehat{\mathbf A}_3\right]_{q,r}e^{j2\pi c_2\bar m_q^2},
	\quad q\in\mathcal Q .
	\label{eq:lipdtft_dechirp}
\end{align}
The integer peak and its stronger adjacent sample are determined by
\begin{align}
	q_r^\star
	&=\underset{q\in\mathcal Q}{\arg\max}\,
	\left|\left[\widehat{\mathbf A}_3\right]_{q,r}\right|,
	\nonumber\\
	q_r^{\mathrm a}
	&=\underset{q\in\{q_r^\star-1,q_r^\star+1\}\cap\mathcal Q}{\arg\max}\,
	\left|\left[\widehat{\mathbf A}_3\right]_{q,r}\right|,
	\nonumber\\
	\delta_{0,r}
	&=(q_r^{\mathrm a}-q_r^\star)
	\frac{\left|\left[\widehat{\mathbf A}_3\right]_{q_r^{\mathrm a},r}\right|}
	{\left|\left[\widehat{\mathbf A}_3\right]_{q_r^{\mathrm a},r}\right|+
	\left|\left[\widehat{\mathbf A}_3\right]_{q_r^\star,r}\right|}.
	\label{eq:lipdtft_initial}
\end{align}
Accordingly, the coarse continuous AF-domain peak estimate is
\begin{equation}
	\widehat m_{0,r}=\bar m_{q_r^\star}+\delta_{0,r}.
	\label{eq:lipdtft_initial_peak}
\end{equation}
Next, we evaluate the recovered AF-domain response at two continuous
points located symmetrically around $\widehat m_{0,r}$. Define the
Nyquist-interpolated continuous response and the normalized complex Dirichlet kernel as
\begin{equation}
    \mathcal S_r(\mu)
    =
    \sum_{q\in\mathcal Q}
    \check a_{3,r}(q)
    D_M(\mu-\bar m_q),
\end{equation}
and
\begin{align}  
    P_{-,r}
    &=
    \left|
    \mathcal S_r(\widehat m_{0,r}-d_x)
	\right|\nonumber
    \\
    &=
    \left|
    \sum_{q\in\mathcal Q}
    \check a_{3,r}(q)
    D_M\!\left(
        \widehat m_{0,r}-d_x-\bar m_q
    \right)
    \right|,
    \\
    P_{+,r}
    &=
    \left|
    \mathcal S_r(\widehat m_{0,r}+d_x)
    \right|\nonumber
    \\
    &=
    \left|
    \sum_{q\in\mathcal Q}
    \check a_{3,r}(q)
    D_M\!\left(
        \widehat m_{0,r}+d_x-\bar m_q
    \right)
    \right|,
    \label{eq:lipdtft_nyquist}
\end{align}
where $P_{-,r}$ and $P_{+,r}$ denote the magnitudes of the
interpolated AF-domain response evaluated at two symmetric
probing points, $\widehat m_{0,r}-d_x$ and
$\widehat m_{0,r}+d_x$, respectively. $D_M(x)= \frac{\sin(\pi x)}{M\sin(\pi x/M)}e^{-j\pi\frac{M-1}{M}x}$ is the Dirichlet kernel, $d_x\in(0,1)$ is the LIpDTFT interpolation offset, which is set to $d_x=0.1$ in this paper.
To obtain the continuous peak correction, define the residual error
of the initial estimate as
\begin{equation}
	\epsilon_r=\widehat m_{0,r}-m_r,
\end{equation}
Around the main lobe, the magnitude of the interpolated response can be
approximated by
\begin{equation}
\left|\mathcal S_r(\mu)\right|
    \simeq
    A_r\widetilde W_M(\mu-m_r),
\end{equation}
where $A_r>0$ absorbs the pilot amplitude, and TT
scaling ambiguity, and
\begin{equation}
	\widetilde W_M(x)=\frac{\sin(\pi x)}
         {M\sin(\pi x/M)}
\end{equation}
is the magnitude of the rectangular-window DTFT inside its main
lobe. Since $\widetilde W_M(x)$ is even, we have
\begin{align}
    P_{-,r}
    &\simeq
    A_r\widetilde W_M(d_x-\epsilon_r),\\
    P_{+,r}
    &\simeq
    A_r\widetilde W_M(d_x+\epsilon_r).
\end{align}
When $|\epsilon_r|\ll d_x$, first-order Taylor expansion around
$d_x$ gives
\begin{equation}
    \frac{P_{-,r}-P_{+,r}}
         {P_{-,r}+P_{+,r}}
    \simeq
    -
    \frac{\widetilde W_M'(d_x)}
         {\widetilde W_M(d_x)}
    \epsilon_r.
\end{equation}
Therefore,
\begin{equation}
    \widehat\epsilon_r
    \simeq
    -
    \frac{\widetilde W_M(d_x)}
         {\widetilde W_M'(d_x)}
    \frac{P_{-,r}-P_{+,r}}
         {P_{-,r}+P_{+,r}}.
\end{equation}
Since $m_r=\widehat m_{0,r}-\epsilon_r$, the LIpDTFT-refined
sub-bin displacement, continuous peak location, and AF-domain
location parameter are respectively given by
\begin{align}
    \widehat\delta_r
    &=
    \delta_{0,r}
    +
    \gamma_M(d_x)
    \frac{P_{-,r}-P_{+,r}}
         {P_{-,r}+P_{+,r}},
    \nonumber\\
    \widehat m_r
    &=
    \bar m_{q_r^\star}
    +
    \widehat\delta_r,
    \nonumber\\
    \widehat{\operatorname{loc}}_r
    &=
    \widehat m_r-m_{\mathrm{pilot}},
    \label{eq:lipdtft_location}
\end{align}
where $\gamma_M(d_x)=
    \frac{\widetilde W_M(d_x)}
         {\widetilde W_M'(d_x)}$. 
Then, according to the coupling relationship in (29), the fractional delay tap is recovered as
\begin{equation}
	\widehat{l_r+\iota_r}
	=\frac{\widehat\nu_rT-\widehat{\operatorname{loc}}_r}{2Mc_1}.
	\label{estdelay}
\end{equation}
}

\begin{algorithm}[t]
	\caption{Proposed Channel Estimation Algorithm}
	\label{alg1}
	\begin{algorithmic}[1]
		\REQUIRE $\widehat{\mathcalbf{Y}} $.
		\STATE Perform TT-SVD by (\ref{A1}), (\ref{A2}) and (\ref{A3}) to obtain the tensor cores $\mathbf{G}_1$, $\mathcalbf{G} _2$ and $\mathbf{G}_3$.
		\STATE Compute the EVD (\ref{EVDG1}) and extract the generators $\widehat{z}_{1,r}=\mathbf{Z}_1\left( r,r \right) /\left| \mathbf{Z}_1\left( r,r \right) \right|$. 
		\STATE Recover the factor matrix $\widehat{\mathbf{A}}_1$ as $\widehat{\mathbf{a}}_{1,r}=\left[ 1,\widehat{z}_{1,r},...,\widehat{z}_{1,r}^{N_{\mathrm{BS}}-1} \right] ^{\mathrm{T}}$. 
		\STATE Compute the multiplication (\ref{G1G2}) and (\ref{pinvAG1G2}), and recover $\widehat{\mathbf{A}}_2$ and $\widehat{\mathbf{M}}_{2}^{-\mathrm{T}}$ by (\ref{a2M2T}).
		\STATE Recover $\widehat{\mathbf{A}}_3$ by (\ref{G3}) .
		\STATE Estimate the AoAs $\left\{ \widehat{\theta}_r \right\} $ and Doppler frequencies $\left\{ \widehat{\nu}_r \right\} $ by (\ref{estAOA}) and (\ref{estDoppler}), respectively.
		\STATE {Estimate $\left\{\widehat{\operatorname{loc}}_r\right\}$ and $\left\{\widehat{l_r+\iota_r}\right\}$ by the Nyquist-interpolated LIpDTFT method in (\ref{eq:lipdtft_dechirp})--(\ref{estdelay}).}
		\STATE Compute the multiplications by (\ref{ASPI}) and estimate the channel gain $\left\{ \widehat{\tilde{\alpha}}_r \right\} $ by (\ref{estgain}).
		\ENSURE channel parameters $\left\{ \widehat{\theta}_r,\widehat{\nu}_r,\widehat{l}_r+\widehat{\iota}_r,\widehat{\tilde{\alpha}} \right\} $ for $r=1,2,...,P$ and reconstructed MIMO-AFDM channel $\widehat{\mathcalbf{H}}$.  
	\end{algorithmic}
\end{algorithm}
\vspace{-0.5cm}

Finally, we examine the estimation of channel gains. Note that $\widehat{\mathbf{A}}_3$ is subject to an ambiguity factor $\bm{{\bm \delta}}$, which should be eliminated first. Based on the previously estimated AoA and Doppler parameters, we firstly reconstruct the factor matrices $\mathbf{A}_1\mathbf{\Pi }$ and $\mathbf{A}_2\mathbf{\Pi }$ with column ambiguities. Then, defining $\mathcalbf{Z} =\mathbf{G}_1\times _{2}^{1}\mathcalbf{G} _2\times _{3}^{2}\mathbf{G}_3$, we have the $\mathbf{A}_3\mathbf{\Pi }$ without scaling ambiguity as follows:
\begin{equation}
	\begin{aligned}
		\left[ \mathcalbf{Z} \right] _3&=\mathbf{A}_3\left( \mathbf{A}_2\odot \mathbf{A}_1 \right) ^{\mathrm{T}}
\\
\mathbf{A}_3\mathbf{S\Pi }&=\left[ \mathcalbf{Z} \right] _3\left( \left( \mathbf{A}_2\mathbf{\Pi }\odot \mathbf{A}_1\mathbf{\Pi } \right) ^{\dagger} \right) ^{\mathrm{T}},\label{ASPI}
	\end{aligned}
\end{equation}
where $\mathbf{S}=\mathrm{Diag}\left( \left[ \tilde{\alpha}_1,\tilde{\alpha}_2,...,\tilde{\alpha}_{P} \right] ^{\mathrm{T}} \right) $. 
The channel gain can be estimated via (\ref{estgain}), which is at the bottom of the next page.
Compared to the method in \cite{tttsp}, which estimates parameters to recover all factor matrices and then obtains the channel gains by computing the pseudo-inverse of their Khatri-Rao product, the approach proposed in (\ref{estgain}) utilizes the strongest observed sample of the AF-domain response of the $r$-th propagation path to estimate the channel gains, offering greater robustness against noise. The proposed algorithm is detailed in Algorithm \ref{alg1}.

\subsection{Complexity Analysis}
In this subsection, the computational complexity of the proposed algorithm is analyzed. The main complexity of the proposed algorithm lies in SVD and pseudoinverse of matrices. Specifically, the TT-SVD process has a computational complexity of $\mathcal{O} \left( N_{\mathrm{BS}}^{2}NM_{\mathrm{region}}+P^2N^2M_{\mathrm{region}} \right) $. The EVD of (\ref{EVDG1}) and the pseudoinverse and multiplication (\ref{pinvAG1G2}) are on the order of the complexity of $\mathcal{O} \left( N_{\mathrm{BS}}^{3} \right) $ and $\mathcal{O} \left( N_{\mathrm{BS}}^{2}P+N_{\mathrm{BS}}P^2N \right) $, respectively.
For each propagation path, the proposed LIpDTFT refinement evaluates two continuous DTFT samples over $M_{\mathrm{region}}$ observed AF-domain entries. Therefore, its additional complexity is $\mathcal{O}(PM_{\mathrm{region}})$, which is negligible compared with the TT-SVD and matrix pseudoinverse operations and does not change the dominant-order complexity of the proposed algorithm.
Finally, the total complexity of the proposed algorithm is of order $\mathcal{O} \left( N_{\mathrm{BS}}^{2}NM_{\mathrm{region}}+P^2N^2M_{\mathrm{region}}+N_{\mathrm{BS}}^{2}P+N_{\mathrm{BS}}P^2N \right) $.

{
\subsection{Identifiability of the Recovered Channel Parameters}
\label{subsec:identifiability}

In this subsection, we analyze the identifiability of the recovered channel parameters.
Although the individual TT cores (\ref{A1})-(\ref{A3}) are not unique \cite{TTD}, the physical
channel parameters recovered by the proposed algorithm are
identifiable under the following conditions. We assume $P$ resolvable paths with nonzero gains and rank-$P$ TT
unfoldings.
First, the spatial generators are recovered by (\ref{EVDG1}). Since
$\mathbf{G}_1=\mathbf{A}_1\mathbf{M}_1^{-1}$ and $\mathbf{M}_1$ is
nonsingular, we have
\[
\operatorname{rank}
\left(
\underline{\mathbf{G}}_1
\right)
=
\operatorname{rank}
\left(
\underline{\mathbf{A}}_1
\right).
\]
Thus, the AoAs can be uniquely identified by (\ref{estAOA}) when
\setcounter{equation}{47}
\begin{equation}
\operatorname{rank}
\left(
\underline{\mathbf{A}}_1
\right)=P.
\label{eq:spatial_identifiability}
\end{equation}

Next, the temporal generators are recovered by (\ref{a2M2T}).
Since
$\widehat{\mathbf{A}}_2
=\mathbf{A}_2\mathbf{\Pi}\mathbf{\Delta}^{-1}$
and $\mathbf{\Delta}$ is nonsingular, the permutation and scaling
ambiguities do not affect the adjacent-entry ratio. Thus, the Doppler
frequencies can be uniquely recovered by (\ref{estDoppler}) within their unambiguous interval
when
\begin{equation}
N\geq2.
\label{eq:temporal_identifiability}
\end{equation}

Then, the AF-domain factor is recovered by (\ref{G3}). Since
$\widehat{\mathbf{A}}_3
=\mathbf{A}_3\mathbf{\Pi}\mathbf{\Delta}$ and
$\mathbf{\Pi}\mathbf{\Delta}$ is nonsingular, all $P$ AF-domain factors
are preserved when
\begin{equation}
\operatorname{rank}
\left(
\mathbf{A}_3
\right)=P.
\label{eq:AF_identifiability}
\end{equation}
The scaling in $\widehat{\mathbf{A}}_3$ does not affect the normalized
magnitude relations used by the LIpDTFT refinement. Therefore, provided
that the AF-domain location mapping is injective over the prescribed
delay--Doppler region and $c_1\neq0$, the continuous AF-domain locations
and fractional delays are uniquely recovered from (\ref{estdelay}).

Finally, the gain-absorbed AF-domain factor is recovered by
(\ref{ASPI}). The pseudoinverse in (\ref{ASPI}) uniquely determines
this factor when
\begin{equation}
\operatorname{rank}
\left(
\mathbf{A}_2\odot\mathbf{A}_1
\right)=P.
\label{eq:gain_identifiability}
\end{equation}
The path gains are then uniquely recovered from (\ref{estgain}).

In summary, the proposed algorithm uniquely recovers the channel
parameters under the following conditions:
\begin{equation}
\begin{aligned}
\operatorname{rank}
\left(
\underline{\mathbf{A}}_1
\right)&=P,
\\
N&\geq2,
\\
\operatorname{rank}
\left(
\mathbf{A}_3
\right)&=P,
\\
\operatorname{rank}
\left(
\mathbf{A}_2\odot\mathbf{A}_1
\right)&=P,
\end{aligned}
\label{eq:overall_identifiability}
\end{equation}
Under these
conditions, the physical channel parameters
\begin{equation}
\left\{
\theta_r,\nu_r,l_r+\iota_r,\widetilde{\alpha}_r
\right\}_{r=1}^{P}
\end{equation}
are uniquely identifiable up to a common path permutation, although
the individual TT cores are not unique.
}

\setcounter{equation}{46}
\begin{figure*}[hb]
	\hrulefill
	\begin{equation}
	\widehat{\tilde{\alpha}}_r=\frac{\left( \mathbf{A}_3\mathbf{S\Pi } \right) _{q_r^\star,r}}{e^{j\frac{2\pi}{M}\left( Mc_1\left( \widehat{l_r+\iota_r} \right) ^2-m_{\mathrm{pilot}}\widehat{l_r+\iota_r}+Mc_2\left( m_{\mathrm{pilot}}^{2}-\bar m_{q_r^\star}^{2} \right) \right)}\frac{e^{-j2\pi \left( \bar m_{q_r^\star}-m_{\mathrm{pilot}}-\widehat{\nu}_rT+2Mc_1\widehat{l_r+\iota_r} \right)}-1}{Me^{-j\frac{2\pi}{M}\left( \bar m_{q_r^\star}-m_{\mathrm{pilot}}-\widehat{\nu}_rT+2Mc_1\widehat{l_r+\iota_r} \right)}-M}}\label{estgain}
\end{equation}
\end{figure*}

\section{ZZB Analysis}\label{ZZBAnalysis}
The lower bound of the MSE is essential for evaluating channel parameter estimation performance. Most existing researches \cite{zhouzhoujsac,nomp} utilized the CRB as the performance benchmark. Nevertheless, at low signal-to-noise ratios (SNRs), the CRB fails to provide a tight bound on the actual estimation performance. As a local non-Bayesian bound, it exhibits asymptotic tightness only under small estimation error conditions (high SNR scenarios). In numerical simulations, the estimation performance generally suffers from severe degradation as the SNR drops in the low-SNR regime. This phenomenon is known as the threshold effect \cite{ZZB,Fundamentals,detI}. In this section, we investigate the closed-form ZZB for estimation of channel parameters in multipath propagation, including AoAs, Doppler frequencies, and delay taps. Without loss of generality, we assume that the signals received over different propagation paths are incoherent, and the channel parameters associated with each path are independently and identically distributed (i.i.d.) while being mutually distinct.

Define $\eta _i=l_i+\iota _i$ and $\bm{p }=\left[ \bm{\theta }^{\mathrm{T}},\bm{\nu }^{\mathrm{T}},\bm{\eta }^{\mathrm{T}} \right] ^{\mathrm{T}}$. The vectorized form of the tensor $\mathcalbf{Y} $ is denoted by $\mathbf{y}$, and is given by
\setcounter{equation}{53}
\begin{align}
	\widehat{\mathbf{y}}&=\mathrm{vec}\left( \widehat{\mathcalbf{Y}} \right) \nonumber
\\
&=\sum_{i=1}^{P}{}\tilde{\alpha}_i\bm{\phi }\left( \theta _i,\nu _i,\eta _i \right) +\mathbf{n}\nonumber
\\
&=\mathbf{\Phi }\left( \bm{\theta },\bm{\nu },\bm{\eta } \right) \bm{\alpha }+\mathbf{n},
\end{align}
where 
\begin{align}
	\bm{\phi }\left( \theta _i,\nu _i,\eta _i \right) &=\mathbf{a}_{N_{\mathrm{BS}}}\left( \theta _i \right) \otimes \mathbf{b}_N\left( \nu _i \right) \otimes \mathbf{c}_{M_{\mathrm{pilot}}}\left( \nu _i,\eta _i \right) \nonumber
\\
\bm{\Phi }\left( \bm{p} \right) &=\left[ \bm{\phi }\left( \theta _1,\nu _1,\eta _1 \right) ,\bm{\phi }\left( \theta _2,\nu _2,\eta _2 \right) ,...,\bm{\phi }\left( \theta _P,\nu _P,\eta _P \right) \right] ,\nonumber
\end{align}
and $\mathbf{\alpha }=\left[ \tilde{\alpha}_1,\tilde{\alpha}_2,...,\tilde{\alpha}_P \right] ^{\mathrm{T}}$ and $\mathbf{n}=\mathrm{vec}\left( \mathcalbf{N} \right) $.
Since DAFT is a unitary transform, each element of the noise tensor $\mathcalbf{N}$ still follows an i.i.d. complex Gaussian distribution $\mathcal{C} \mathcal{N} \left( 0,\sigma _{\mathrm{n}}^{2} \right) $. 
Given the channel parameters $\bm{p }$, the covariance matrix $\mathbf{R}_{\mathbf{y}|\bm{p }}$ is written as
\begin{align}
		\mathbf{R}_{\mathbf{y}|\bm{p }}&=\mathbb{E} \left\{ \mathbf{yy}^{\mathrm{H}} \right\} \nonumber
\\
&=\mathbf{\Phi }\left( \bm{p } \right) \mathbf{\Sigma \Phi }^{\mathrm{H}}\left( \bm{p } \right) +\sigma _{n}^{2}\mathbf{I}_{N_{\mathrm{BS}}NM_{\mathrm{region}}},
\end{align}
where $\mathbf{\Sigma }=\mathbb{E} \left\{ \bm{\alpha \alpha }^{\mathrm{H}} \right\} =\mathrm{diag}\left( \left[ \sigma _{1,\mathrm{PL}}^{2},\sigma _{2,\mathrm{PL}}^{2},...,\sigma _{P,\mathrm{PL}}^{2} \right] ^{\mathrm{T}} \right) $. The error correlation matrix is defined as
\begin{equation}
	\mathbf{R}_{\mathbf{\epsilon }}=\mathbb{E} \left\{ \mathbf{\epsilon \epsilon }^{\mathrm{T}} \right\} =\mathbb{E} \left\{ \left( \hat{\bm{p }}-\bm{p } \right) \left( \hat{\bm{p }}-\bm{p } \right) ^{\mathrm{T}} \right\}.
\end{equation}
Based on the above definitions, the extended ZZB is given by \cite{exzzb,zzb2d}
\begin{align}
		\bm{w}^{\mathrm{T}}\mathbf{R}_{\bm{\epsilon }}\bm{w}&\ge \frac{1}{2}\int_0^{\infty}{}\max_{\bm{{\bm \delta} }:\bm{w}^{\mathrm{T}}\bm{{\bm \delta} }=h} \left[ \int_{\mathbb{R} ^{3P}}^{}{}\left( f_{\bm{p }}\left( \bm{{\bm \varphi} } \right) +f_{\bm{p }}\left( \bm{{\bm \varphi} }+\bm{{\bm \delta} } \right) \right) \right.\nonumber
		\\
		&\quad\times \left.P_{\min}\left( \bm{{\bm \varphi} },\bm{{\bm \varphi} }+\bm{{\bm \delta} } \right) \mathrm{d}\bm{{\bm \varphi} } \right] h\mathrm{d}h,\label{wrw}
\end{align}
where the normalized weight vector $\bm{w}\in \mathbb{R} ^{3P}$, satisfying $\left\| \bm{w} \right\| _2=1$. $f_{\bm{p}}\left( \bm{{\bm \varphi} } \right) $ and $f_{\bm{p}}\left( \bm{{\bm \varphi} }+\bm{{\bm \delta} } \right) $ are the a priori probability density function (pdf) of $\bm{p}$ given $\bm{p}=\bm{{\bm \varphi} }$ and $\bm{p}=\bm{{\bm \varphi} }+\bm{{\bm \delta} }$, respectively, and $P_{\min}\left( \bm{{\bm \varphi} },\bm{{\bm \varphi} }+\bm{{\bm \delta} } \right) $ represents the minimum probability of error associated with the following binary hypothesis testing problem:
\begin{align}
	&\mathcal{H} _0:\bm{p}=\bm{{\bm \varphi} }; &&\bm{y}\sim f\left( \bm{y}| \bm{p}=\bm{{\bm \varphi} } \right),
	\nonumber
\\
&\mathcal{H} _1:\bm{p}=\bm{{\bm \varphi} }+\bm{{\bm \delta} };&& \bm{y}\sim f\left( \bm{y}| \bm{p}=\bm{{\bm \varphi} }+\bm{{\bm \delta} } \right),
\end{align}
where $f\left( \mathbf{y}|\bm{p}=\bm{{\bm \varphi}} \right) $ and $f\left( \mathbf{y}|\bm{p}=\bm{{\bm \varphi}}+\bm{{\bm \delta}} \right) $ are the conditional PDF of $\mathbf{y}$ given $\bm{p}=\bm{{\bm \varphi}}$ and $\bm{p}=\bm{{\bm \varphi}}+\bm{{\bm \delta}}$, respectively, and the prior probabilities of the two hypotheses are given by
\begin{align}
	\mathrm{Pr}\left( \mathcal{H} _0 \right) &=\frac{f_{\bm{p}}\left( \bm{{\bm \varphi} } \right)}{f_{\bm{p}}\left( \bm{{\bm \varphi} }+\bm{{\bm \delta} } \right) +f_{\bm{p}}\left( \bm{{\bm \varphi} } \right)},\nonumber
\\
\mathrm{Pr}\left( \mathcal{H} _1 \right) &=1-\mathrm{Pr}\left( \mathcal{H} _0 \right).\label{MSELB}
\end{align}
Then, the MSE of the $k$-th set of channel parameters in ${\bm{p}}$ can be lower bounded by
\begin{equation}
	\begin{aligned}
		\mathrm{MSE}\left( \bm{p}_k \right) &=\bm{w}_{k}^{\mathrm{T}}\mathbf{R}_{\bm{\epsilon }}\bm{w}_k
\\
&\ge \frac{1}{2}\int_0^{\infty}{}\max_{\bm{{\bm \delta}}:\bm{w}_{k}^{\mathrm{T}}\bm{{\bm \delta}}=h} \left[ \int_{\mathbb{R} ^{3P}}^{}{}\left( f_{\bm{p}}\left( \bm{{\bm \varphi}} \right) +f_{\bm{p}}\left( \bm{{\bm \varphi}}+\bm{{\bm \delta}} \right) \right) \right.
\\
&\quad\left.\times P_{\min}\left( \bm{{\bm \varphi}},\bm{{\bm \varphi}}+\bm{{\bm \delta}} \right) \mathrm{d}\bm{{\bm \varphi}} \right] h\mathrm{d}h,
	\end{aligned}
\end{equation}
where
\begin{equation}
\bm{p}_k=\bm{p}\left( \left( k-1 \right) P+1:kP \right) =\begin{cases}
	\bm{\theta }, k=1\\
	\bm{\nu }, k=2\\
	\bm{\eta }, k=3,\\
\end{cases}
\end{equation}
and 
\begin{equation}
	\bm{w}_k=\frac{1}{\sqrt{P}}\left[ \mathbf{0}_{\left( k-1 \right) P}^{\mathrm{T}},\mathbf{1}_{P}^{\mathrm{T}},\mathbf{0}_{\left( 3-k \right) P}^{\mathrm{T}} \right] ^{\mathrm{T}}.
\end{equation}
However, the expression of $P_{\min}\left( \bm{{\bm \varphi}},\bm{{\bm \varphi}}+\bm{{\bm \delta}} \right) \mathrm{d}\bm{{\bm \varphi}}$ is generally complicated, and directly substituting it into (\ref{MSELB}) usually makes it difficult to obtain an explicit closed-form expression. Fortunately, a lower bound for $P_{\min}\left( \bm{{\bm \varphi}},\bm{{\bm \varphi}}+\bm{{\bm \delta}} \right) \mathrm{d}\bm{{\bm \varphi}}$ for observation $\bm{y}$ is given by \cite{detIII}
\begin{align}\label{Pdelta}
		&P_{\min}\left( \bm{{\bm \varphi} },\bm{{\bm \varphi} }+\bm{{\bm \delta} } \right) \nonumber
		\\
		&\quad\ge P\left( \bm{{\bm \delta} } \right) \nonumber
\\
&\quad=e^{\left[ \mu \left( s;\bm{{\bm \delta} } \right) +\frac{1}{8}\frac{\partial ^2\mu \left( s;\bm{{\bm \delta} } \right)}{\partial s^2} \right]}\left.\mathcal{Q} \left( \frac{1}{2}\sqrt{\frac{\partial ^2\mu \left( s;\bm{{\bm \delta} } \right)}{\partial s^2}} \right) \right|_{s=\frac{1}{2}}^{},
\end{align}
where $\mathcal{Q} \left( x \right) =\int_x^{\infty}{}\frac{1}{\sqrt{2\pi}}e^{-\frac{v^2}{2}}\mathrm{d}v$ is the Q-function and 
\begin{align}
		\mu \left( s;\mathbf{{\bm \delta} } \right) &=\ln \int{}f\left( \mathbf{y}|\mathbf{{\bm \varphi} }+\mathbf{{\bm \delta} } \right) ^sf\left( \mathbf{y}|\mathbf{{\bm \varphi} } \right) ^{1-s}\mathrm{d}\mathbf{y}\nonumber
\\
&=\ln \int{}\frac{1}{\pi ^{N_{\mathrm{BS}}NM_{\mathrm{pilot}}}\left| \mathbf{R}_{\mathbf{y}|\mathbf{{\bm \varphi} }+\mathbf{{\bm \delta} }} \right|^s\left| \mathbf{R}_{\mathbf{y}|\mathbf{{\bm \varphi} }} \right|^{1-s}}\nonumber
\\
&\quad\times \exp \left( -\mathbf{y}^{\mathrm{H}}\left( s\mathbf{R}_{\mathbf{y}|\mathbf{{\bm \varphi} }+\mathbf{{\bm \delta} }}^{-1}+\left( 1-s \right) \mathbf{R}_{\mathbf{y}|\mathbf{{\bm \varphi} }}^{-1} \right) \mathbf{y} \right) \mathrm{d}\mathbf{y}\nonumber
\\
&=-\ln \left| s\mathbf{R}_{\mathbf{y}|\mathbf{{\bm \varphi} }+\mathbf{{\bm \delta} }}^{-1}+\left( 1-s \right) \mathbf{R}_{\mathbf{y}|\mathbf{{\bm \varphi} }}^{-1} \right|\nonumber
\\
&\quad-s\ln \left| \mathbf{R}_{\mathbf{y}|\mathbf{{\bm \varphi} }+\mathbf{{\bm \delta} }} \right|-\left( 1-s \right) \ln \left| \mathbf{R}_{\mathbf{y}|\mathbf{{\bm \varphi} }} \right|\label{CGF}
\end{align}
is the cumulant generating function (CGF) of the log-likelihood ratio $\ln \frac{f\left( \mathbf{y}|\mathbf{{\bm \varphi} }+\mathbf{{\bm \delta} } \right)}{f\left( \mathbf{y}|\mathbf{{\bm \varphi} } \right)}$. We present the following theorem that gives the closed-form expression for $P\left( \mathbf{{\bm \delta} } \right) $:
\begin{Theorem}\label{theorem2}
	Based on the CGF given in (\ref{CGF}), we have 
	\begin{equation}\label{Pdeltaexpression}
		P\left( \mathbf{{\bm \delta} } \right) \approx \begin{cases}
	\mathcal{Q} \left( \frac{1}{2}\sqrt{\mathbf{{\bm \delta} }^{\mathrm{T}}\mathbf{J{\bm \delta} }} \right) , &\mathbf{{\bm \delta} }\in \mathbf{\Delta },\\
	P_{\mathrm{NA}}\left( \mathbf{\bm \rho } \right) ,&\mathbf{{\bm \delta} }\notin \mathbf{\Delta },\\
\end{cases}
	\end{equation}
where $\left[ \mathbf{J} \right] _{ij}=\mathrm{Tr}\left( \mathbf{R}^{-1}\frac{\partial \mathbf{R}}{\partial p_i}\mathbf{R}^{-1}\frac{\partial \mathbf{R}}{\partial p_j} \right) $ is the Fisher matrix and 
\begin{equation}\label{PNA}
	P_{\mathrm{NA}}=e^{\psi \left( \mathbf{\bm \rho } \right)}\mathcal{Q} \left( \sqrt{\frac{1}{2}\sum_{i=0}^{P-1}{}\left( \frac{2N_{\mathrm{BS}}N\rho _i}{2+N_{\mathrm{BS}}N\rho _i} \right) ^2} \right), 
\end{equation}
where 
\begin{equation}
	\psi \left( \mathbf{\bm \rho } \right) =\sum_{i=0}^{P-1}\ln \left[ \frac{4\left( 1+N_{\mathrm{BS}}N\rho _i \right)}{\left( 2+N_{\mathrm{BS}}N\rho _i \right) ^2} \right] +\left( \frac{N_{\mathrm{BS}}N\rho _i}{2+N_{\mathrm{BS}}N\rho _i} \right) ^2,
\end{equation}
and the $i$-th element of $\mathbf{\bm \rho }$ is the SNR of the $i$-th propagation path, which is defined as $\rho _i=\frac{\sigma _{i,\mathrm{PL}}^{2}}{\sigma _{n}^{2}}$. The asymptotic region is defined as $\mathbf{\Delta }=\left\{ \mathbf{{\bm \delta} }|\mathbf{{\bm \delta} }^{\mathrm{T}}\mathbf{J{\bm \delta} }\ge 2\sum_{i=1}^P{}\left( \frac{2N_{\mathrm{BS}}N\rho _i}{2+N_{\mathrm{BS}}N\rho _i} \right) ^2 \right\} $. 

\end{Theorem}

\textit{Proof:} The proof is provided in Appendix \ref{proofoftheorem2}. \hfill $\blacksquare$

Based on Theorem \ref{theorem2}, (\ref{Pdelta}) is lower bounded by
\begin{equation}
	\begin{aligned}
		&\frac{1}{2}\int_{\mathbb{R} ^{3P}}^{}{}\left( f_{\bm{p}}\left( \mathbf{{\bm \varphi} } \right) +f_{\bm{p}}\left( \mathbf{{\bm \varphi} }+\mathbf{{\bm \delta} } \right) \right) P_{\min}\left( \mathbf{{\bm \varphi} },\mathbf{{\bm \varphi} }+\mathbf{{\bm \delta} } \right) \mathrm{d}\mathbf{{\bm \varphi} }
		\\
		&\quad\ge P\left( \mathbf{{\bm \delta} } \right) \prod_{k=1}^3{}\prod_{i=1}^P{}\frac{\left( \zeta _k-{\bm \delta} _{k,i} \right)}{\zeta _k},
	\end{aligned}
\end{equation}
where $\zeta _k$ is the range of the uniform distribution for the $k$-th group of channel parameters in $\bm{p}$. Hence, (\ref{wrw}) reduces to
\begin{equation}\label{newmax}
	\bm{w}^{\mathrm{T}}\mathbf R_{\bm\epsilon}\bm w\ge \max_{\mathbf{{\bm \delta} }:\bm{w}_{k}^{\mathrm{T}}\mathbf{{\bm \delta} }=h} P\left( \mathbf{{\bm \delta} } \right) \prod_{k=1}^3{}\prod_{i=1}^P{}\frac{\left( \zeta _k-{\bm \delta} _{k,i} \right)}{\zeta _k}.
\end{equation}
Substituting (\ref{Pdeltaexpression}) into (\ref{newmax}), we relax the maximization problem in (\ref{wrw}) as \cite{zzb1d}
\begin{equation}
	\begin{aligned}
		&\max_{\mathbf{{\bm \delta} }:\bm{w}_{k}^{\mathrm{T}}\mathbf{{\bm \delta} }=h} \frac{P\left( \mathbf{{\bm \delta} } \right)}{\prod_{k=1}^3{}\zeta _{k}^{P}}\prod_{k=1}^3{}\prod_{i=1}^P{}\left( \zeta _k-{\bm \delta} _{k,i} \right) 
\\
&\approx \max_{\mathbf{{\bm \delta} }:\bm{w}_{k}^{\mathrm{T}}\mathbf{{\bm \delta} }=h} \frac{P_{\mathrm{NA}}}{\prod_{k=1}^3{}\zeta _{k}^{P}}\prod_{k=1}^3{}\prod_{i=1}^P{\left( \zeta _k-{\bm \delta} _{k,i} \right)}
\\
&\quad+\max_{\mathbf{{\bm \delta} }\in \Delta :\bm{w}_{k}^{\mathrm{T}}\mathbf{{\bm \delta} }=h} P_{\mathrm{A}}\left( \mathbf{{\bm \delta} } \right) -P_{\mathrm{NA}}.\label{maxtwoparts}
	\end{aligned}
\end{equation}
The second term on the right-hand side of the approximate equality is due to the fact that $\mathbf{{\bm \delta} }\rightarrow \mathbf{0}$ in the asymptotic region. According to the inequality of arithmetic and geometric means, we have the closed-form for the first maximization term as
\begin{equation}
	\begin{aligned}
		&\max_{\mathbf{{\bm \delta} }:\bm{w}_{k}^{\mathrm{T}}\mathbf{{\bm \delta} }=h} \frac{P\left( \mathbf{{\bm \delta} } \right)}{\prod_{k=1}^3{}\zeta _{k}^{P}}\prod_{k=1}^3{}\prod_{i=1}^P{}\left( \zeta _k-{\bm \delta} _{k,i} \right) 
\\
&\quad=P_{\mathrm{NA}}\prod_{k=1}^3{}\left( 1-\frac{h}{\sqrt{P}\zeta _k} \right) ^P.
	\end{aligned}
\end{equation}
The second maximization term is given by
\begin{equation}\label{PA}
	\max_{\mathbf{{\bm \delta} }\in \Delta :\bm{w}_{k}^{\mathrm{T}}\mathbf{{\bm \delta} }=h} P_{\mathrm{A}}\left( \mathbf{{\bm \delta} } \right) =\mathcal{Q} \left( \frac{1}{2\sqrt{\bm{w}_{k}^{\mathrm{T}}\mathbf{J}^{-1}\bm{w}_k}}h \right),
\end{equation}
where the maximum value is achieved at $\mathbf{{\bm \delta} }=\frac{\mathbf{J}^{-1}\bm{w}_k}{\bm{w}_{k}^{\mathrm{T}}\mathbf{J}^{-1}\bm{w}_k}h$.
Substituting (\ref{maxtwoparts})-(\ref{PNA}) into (\ref{newmax}), we obtain an approximate closed-form expression for ZZB as
\begin{align}\label{zzb}
		\bm{w}_{k}^{\mathrm{T}}\mathbf{R}_{\mathbf{\epsilon }}\bm{w}_k&\ge P_{\mathrm{NA}}\int_0^{\sqrt{P}\zeta}{}\left( 1-\frac{h}{\sqrt{P}\zeta _k} \right) ^Ph\mathrm{d}h\nonumber
\\
&\quad+\int_0^{\tilde{h}}{}\left[ \mathcal{Q} \left( \frac{1}{2\sqrt{\bm{w}_{k}^{\mathrm{T}}\mathbf{J}^{-1}\bm{w}_k}}h \right) -P_{\mathrm{NA}} \right] h\mathrm{d}h\nonumber
\\
&\approx P_{\mathrm{NA}}\frac{12\bm{w}_{k}^{\mathrm{T}}\mathbf{R}_{\bm{p}}\bm{w}_k}{\left( P+1 \right) \left( P+2 \right)}\nonumber
\\
&\quad+\int_0^{\tilde{h}}{}\frac{1}{\sqrt{2\pi}}e^{-\frac{h^2}{8\bm{w}_{k}^{\mathrm{T}}\mathbf{J}^{-1}\bm{w}_k}}\frac{h^2}{4\sqrt{\bm{w}_{k}^{\mathrm{T}}\mathbf{J}^{-1}\bm{w}_k}}\mathrm{d}h\nonumber
\\
&=P_{\mathrm{NA}}\frac{12\bm{w}_{k}^{\mathrm{T}}\mathbf{R}_{\bm{p}}\bm{w}_k}{\left( P+1 \right) \left( P+2 \right)}\nonumber
\\
&\quad+\bm{w}_{k}^{\mathrm{T}}\mathbf{J}^{-1}\bm{w}_k\Gamma _{\frac{3}{2}}\left( \frac{\tilde{h}^2}{8\bm{w}_{k}^{\mathrm{T}}\mathbf{J}^{-1}\bm{w}_k} \right),
\end{align}
where the approximate equality is due to the tight upper bound of the Q-function, $\tilde{h}_k$ represents the boundary between the asymptotic region and the non-asymptotic region, which is given by
{
\begin{equation}
	\tilde{h}_{k}^{2}\le 2\bm{w}_{k}^{\mathrm{T}}\mathbf{J}^{-1}\bm{w}_k\sum_{i=0}^{P-1}{}\left( \frac{2N_{\mathrm{BS}}N\rho _i}{2+N_{\mathrm{BS}}N\rho _i} \right) ^2,
\end{equation}
}
with
\begin{equation}
	\tilde{h}_{k}^{2}=\min \left\{ 2\bm{w}_{k}^{\mathrm{T}}\mathbf{J}^{-1}\bm{w}_k\sum_{i=0}^{P-1}{}\left( \frac{2N_{\mathrm{BS}}N\rho _i}{2+N_{\mathrm{BS}}N\rho _i} \right) ^2,K\zeta _{k}^{2} \right\} ,
\end{equation}
and $\Gamma _a\left( \cdot \right) $ represents the incomplete gamma function with parameter $a$.

As can be observed from  (\ref{zzb}), the ZZB constitutes a global lower bound. In the regime of small $\mathbf{{\bm \delta} }$, the ZZB is dominated by the asymptotic CRB. As the SNR $\rho _i$ decreases, the non-asymptotic component $P_{\mathrm{NA}}$ of (\ref{zzb}) increasingly dominates and converges to the parameter covariance.

\section{Simulation Results}\label{SimulationResults}

In this section, simulation results are demonstrated to examine the performance of the proposed channel estimation methods and the comparison with benchmarks. Table \ref{simparameters} gives the simulation parameters. 
The SNR is defined as $\mathrm{SNR}\triangleq \frac{\left\| \mathcalbf{Y} \right\| _{\mathrm{F}}^{2}}{\left\| \mathcalbf{N} \right\| _{\mathrm{F}}^{2}}$. {The maximum MS velocities are set to 300 and 600 kph, resulting in maximum Doppler frequencies of $\nu _{\max}=$ 4167 and 8333 Hz.} As the maximum Doppler tap is $\alpha _{\max}=\nu _{\max}T=0.14$, we set $c_1=\frac{1}{2M}$ to reduce the pilot overhead.

\begin{table}[ht]
	\centering
	
	\caption{Simulation Parameters}
	\label{simparameters}
	\begin{tabular}{|c|c|}
		\hline
		\textbf{Parameter} & \textbf{Values} \\ 
		\hline
		Carrier frequency (GHz) & 15 \\
		\hline
		Subcarrier spacing (kHz) & 30 \\
		\hline
		\makecell[c]{Total number of subcarriers  and\\ number of symbols for pilot \((M, N)\)}  & (1024,10) \\
		\hline
		Length of CPP \(M_{\text{CPP}}\) & 72 \\
		\hline
		Number of BS antennas  & 16 \\
		\hline
		Number of antenna spacing & ${\lambda}/{2}$	\\
		\hline
		Number of MS antennas & 1 \\
		\hline
		The number of propagation paths & 5 \\
		\hline
		Maximum MS velocity (kph) & 300/ 600 \\
		\hline
		Symbol Modulation Scheme & \makecell[c]{16/ 64/\\ 256 QAM} \\
		\hline
	\end{tabular}
	\vspace{-0.5cm}
\end{table}

\subsection{Channel Parameters Estimate Performance}

\begin{figure*}[t]
	\subfigure[AoA.]{
	\begin{minipage}[t]{0.23\linewidth}
	\centering
	\includegraphics[width=1.0\linewidth]{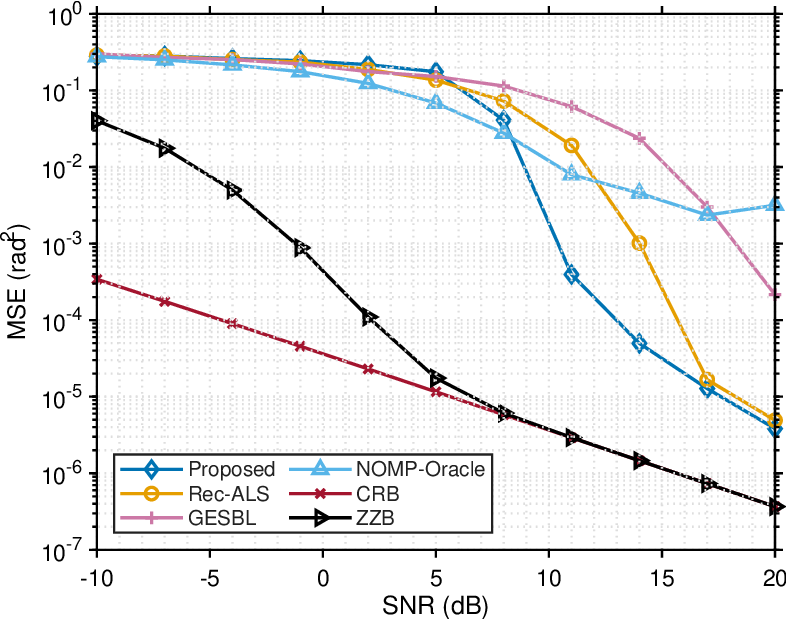}
	\label{MSEAOA}
	\end{minipage}%
	\hfill
	}
	\subfigure[Doppler.]{
	\begin{minipage}[t]{0.23\linewidth}
	\centering
	\includegraphics[width=1.0\linewidth]{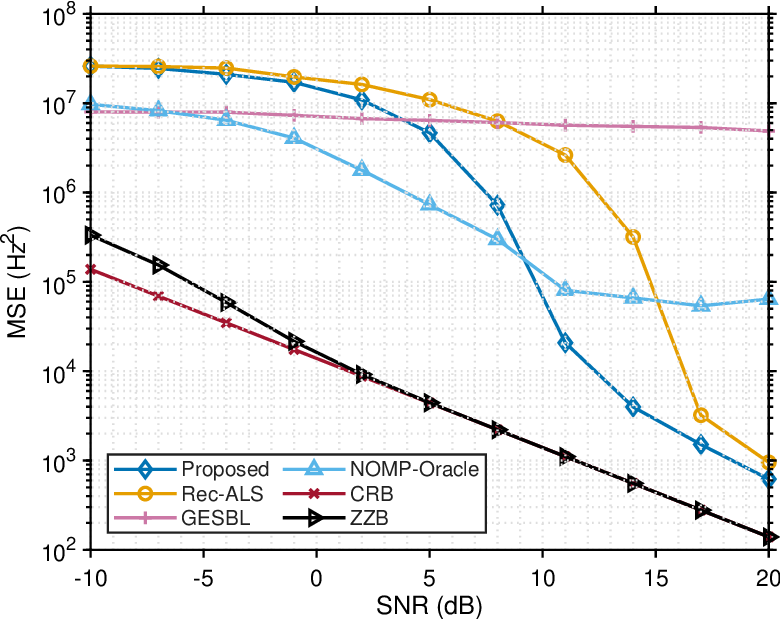}
	\label{MSEDoppler}
	\end{minipage}%
	\hfill
	}
	\subfigure[Delay tap.]{
	\begin{minipage}[t]{0.23\linewidth}
	\centering
	\includegraphics[width=1.0\linewidth]{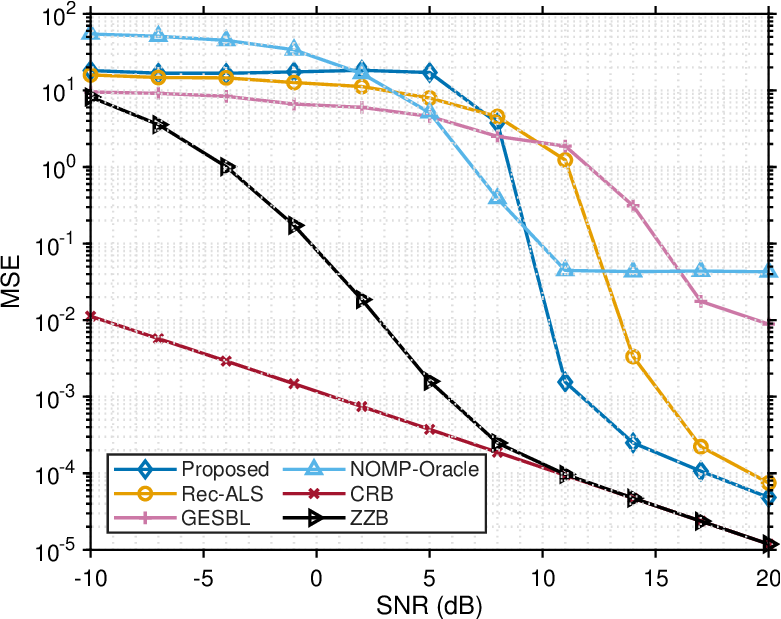}
	\label{MSEDelay}
	\end{minipage}%
	\hfill
	}
	\subfigure[Channel gain.]{
	\begin{minipage}[t]{0.23\linewidth}
	\centering
	\includegraphics[width=1.0\linewidth]{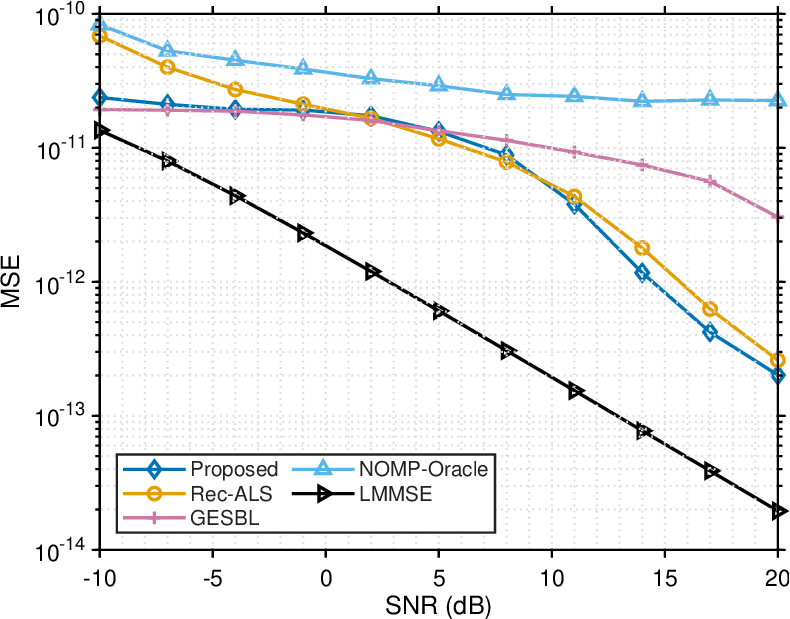}
	\label{MSEgain}
	\end{minipage}%
	\hfill
	}
	\caption{ The MSEs performance versus SNR.}
	\label{MSE}
	\vspace{-0.5cm}
\end{figure*}

\begin{figure}[ht]
	\centering
	\begin{minipage}[htbp]{1.0\linewidth}
	\centering
	\includegraphics[width=1.0\linewidth]{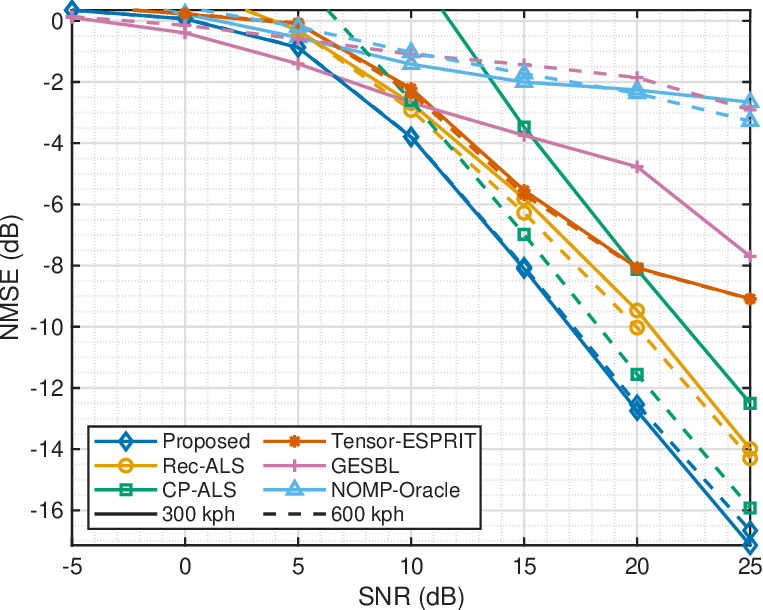}
	\caption{The NMSE performance versus SNR.}
	\label{NMSEvsSNR}
	\end{minipage}%
		\vspace{-0.5cm}
	\hfill
\end{figure}

Fig.~\ref{MSE} shows the MSE performance versus SNR. The total number of subcarriers is set to $M=512$. The length of CPP is set to $M_{\text{CPP}}=45$. {The RecALS \cite{RecALS}, the grid-evolution SBL (GESBL) \cite{GESBL}, and the Newtonized orthogonal matching pursuit (NOMP) \cite{nomp} are selected as benchmarks, which are denoted as ``Rec-ALS'', ``GESBL'', and ``NOMP'', respectively.}
Specifically, the CRB and ZZB serve as benchmarks for AoA, Doppler, and delay estimation, while the LMMSE is used for channel gain. 
It can be observed that in the low SNR regime (SNR $<$ 5 dB), the ZZBs for AoA and delay effectively capture the rapid degradation of estimation error. 
When SNR $=$ -10 dB, the MSE of the proposed algorithm is comparable to the ZZB.
However, in the SNR range of -5 dB to 5 dB, a noticeable gap exists between the MSE of the proposed algorithm and the ZZB. 
This discrepancy stems from the fact  that the ZZB remains a performance bound derived from the  ML criterion \cite{ZZB}, and its behavior in the non-asymptotic regime is predominantly governed by the array gain $N_{\mathrm{BS}}N\rho _i$.
While the proposed subspace-based algorithm for off-grid parameter estimation offers significantly higher computational efficiency, it generally yields suboptimal performance relative to ML estimation (MLE).
Furthermore, the stochastic interaction between MSE for subspace-based method and noise is complicated and cannot be straightforwardly formulated as a function of array gain $N_{\mathrm{BS}}N\rho _i$.
In addition, due to the relatively short symbol duration $T_{\mathrm{sym}}$, the MSE of Doppler estimation in the low-SNR regime ($\leq$ 5 dB) is considerably larger than both the ZZB and the CRB.
While the SNR $\geq$ 10 dB, the ZZBs asymptotically converges to CRB. The MSEs of the proposed algorithm for AoA, Doppler and delay become close to ZZB and CRB.
Finally, the MSE of the channel gain exhibits a similar trend with respect to the LMMSE: it remains relatively close in both the low SNR and medium-to-high SNR regions, while showing a larger gap in the low-to-medium SNR range.

\subsection{Communication Performance Comparison with Existing Methods}

In this subsection, we examine the communication performance for the proposed channel estimation algorithms and the benchmarks.
The NMSE performance is first examined. The NMSE is defined as 
\begin{equation}
	\mathrm{NMSE}\triangleq \mathbb{E} \left\{ \frac{\left\| \hat{\mathcalbf{H}}^{\mathrm{STAF}}-\mathcalbf{H} ^{\mathrm{STAF}} \right\| _{F}^{2}}{\left\| \mathcalbf{H} ^{\mathrm{STAF}} \right\| _{F}^{2}} \right\},
\end{equation}
where $\mathcalbf{H} ^{\mathrm{STAF}}\in \mathbb{C} ^{N_{\mathrm{BS}}\times N\times M}$ and $\hat{\mathcalbf{H}}^{\mathrm{STAF}}\in \mathbb{C} ^{N_{\mathrm{BS}}\times N\times M}$ are the actual STAF channel and the estimated STAF channel, respectively.

{
Fig.~\ref{NMSEvsSNR} illustrates the NMSE performance for the proposed methods and benchmarks.
The CP-ALS \cite{zhouzhoujsac}, the RecALS, the GESBL , the tensor-ESPRIT \cite{tensorESPRIT} and NOMP are selected as benchmarks.
From Fig.~\ref{NMSEvsSNR}, it can be observed that the proposed method achieves better accurate channel estimation performance than the benchmarks. 
The NMSE performance of the RecALS method is also relatively close to that of the proposed method. 
The main reason is both the proposed method and the RecALS method effectively exploit the rotation-invariant properties in both the spatial and time domains. 
Nevertheless, the inherently weak convergence of alternative LS process limits the average NMSE performance of RecALS.
In addition, the NOMP algorithm based on \cite{nomp} performs unsatisfactorily in multipath channel estimation. 
The reason is that the NOMP algorithm discards the off-diagonal blocks of the Hessian matrix during global iterations, which means the algorithm essentially still performs local iterations for each path separately. 
As a result, the iteration process may incurs significant bias.
Furthermore, it is worth notice that, as the velocity increases from 300 to 600 kph, the enlarged Doppler separation substantially improves the numerical conditioning and estimation performance of unconstrained CP-ALS, whereas RecALS exhibits only marginal gains because its Vandermonde rectification already mitigates the correlation among the temporal factors.
}

\begin{figure*}[t]
	\centering
	\subfigure[16 QAM.]{
	\begin{minipage}[t]{0.3\linewidth}
	\centering
	\includegraphics[width=1.0\linewidth]{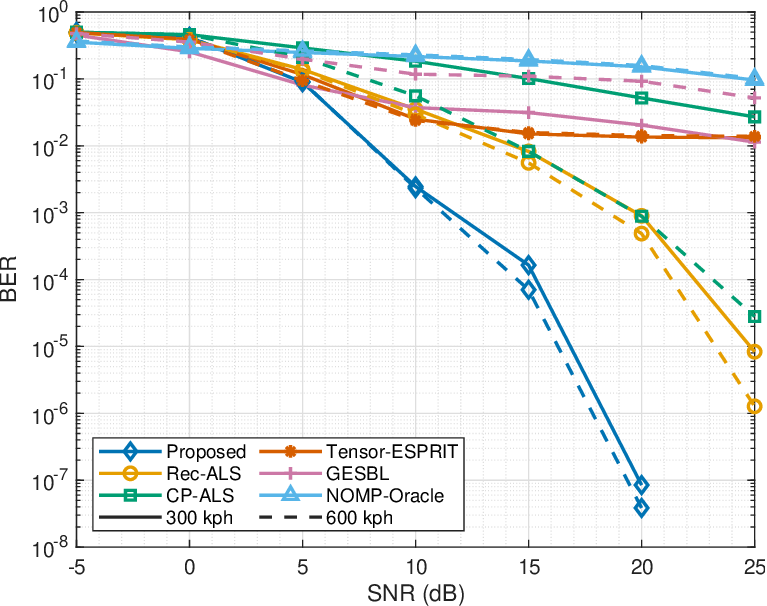}
	\label{BERvsSNR16QAM}
	\end{minipage}%
	\hfill
	}
	\subfigure[64 QAM.]{
	\begin{minipage}[t]{0.3\linewidth}
	\centering
	\includegraphics[width=1.0\linewidth]{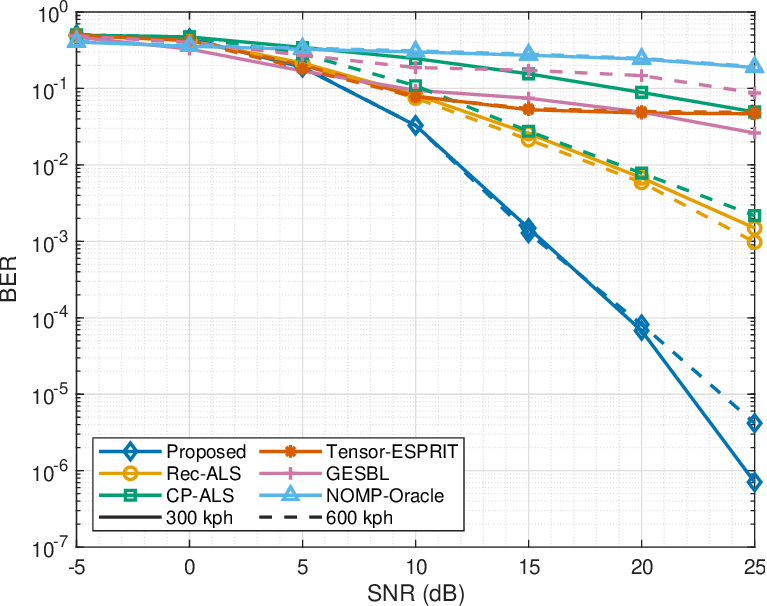}
	\label{BERvsSNR64QAM}
	\end{minipage}%
	\hfill
	}
	\subfigure[256 QAM.]{
	\begin{minipage}[t]{0.3\linewidth}
	\centering
	\includegraphics[width=1.0\linewidth]{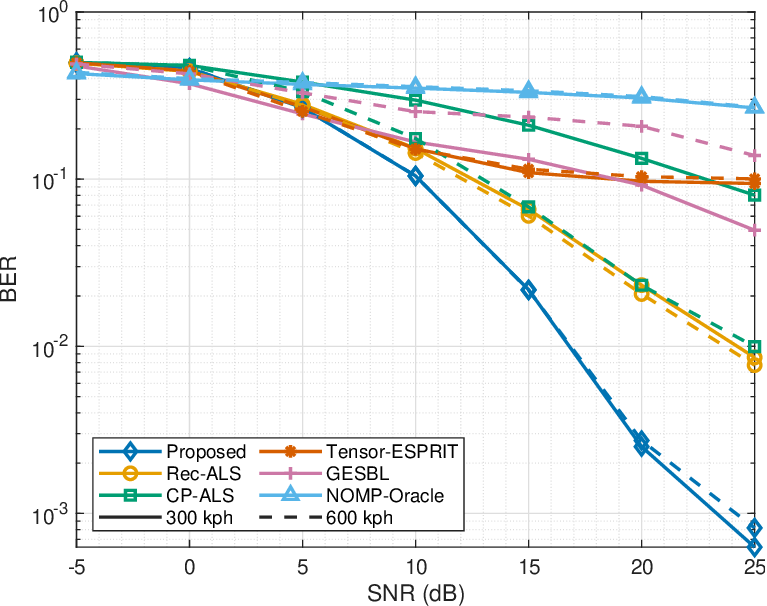}
	\label{BERvsSNR256QAM}
	\end{minipage}%
	\hfill
	}
	\caption{The BER performance versus SNR for different order QAM modulations.}
	\label{BER}
	\vspace{-0.5cm}
\end{figure*}

{
The BER performance versus SNR for the proposed method and benchmarks for different order QAM modulations is demonstrated in Fig.~\ref{BER}. 
From Fig.~\ref{BER}, the BER performance of the proposed method outperforms that of the benchmarks among all the three QAM modulation orders.In particular, at high SNRs, the BER of the proposed method falls below $10^{-4}$ for 64-QAM and $10^{-2}$ for 256-QAM, respectively, demonstrating its potential for high-order QAM transmission in MIMO-AFDM systems.}

\begin{figure}[t]
	\centering
	\begin{minipage}[htbp]{1.0\linewidth}
	\centering
	\includegraphics[width=1.0\linewidth]{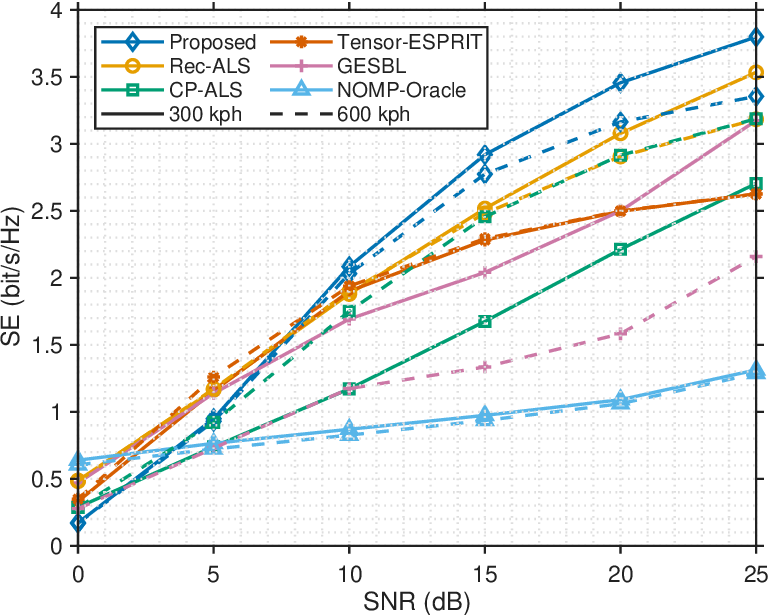}
	\caption{The SE performance versus SNR.}
	\label{SE}
	\end{minipage}%
	\vspace{-0.5cm}
	\hfill
\end{figure}

Fig.~\ref{SE} illustrates the SE performance versus SNR. The proposed method achieves SE performance comparable to that of the RecALS algorithm. In the SNR range of 0--5 dB, RecALS exhibits slightly better SE performance. However, when the SNR exceeds 5 dB, the proposed method achieves higher SE because its parameter-estimation MSE decreases rapidly with increasing SNR and gradually approaches the CRB, resulting in more accurate channel reconstruction and, consequently, improved SE performance.

\begin{figure}[t]
	\centering
	\begin{minipage}[ht]{1.0\linewidth}
	\centering
	\includegraphics[width=1.0\linewidth]{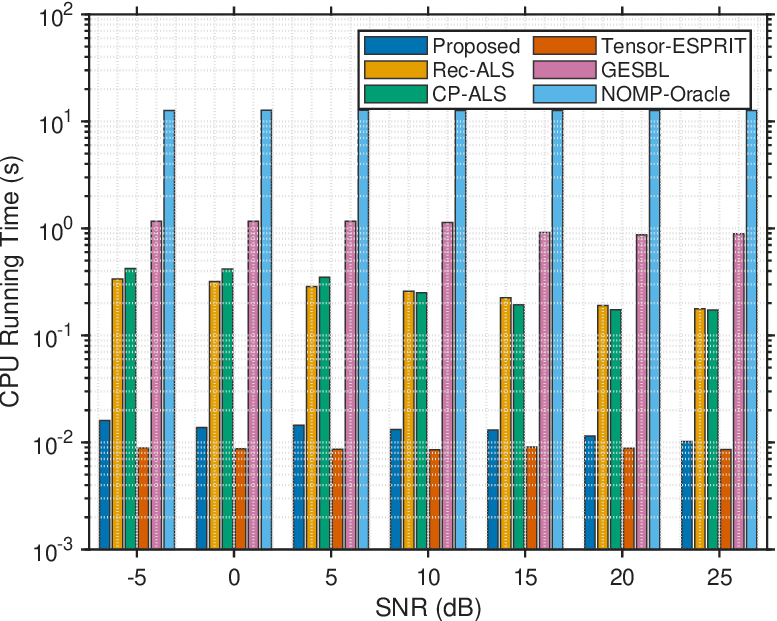}
	\caption{The average CPU running time versus SNR.}
	\label{Time}
	\end{minipage}%
	\vspace{-0.5cm}
	\hfill
\end{figure}

{To evaluate the computational efficiency of different AFDM channel estimation algorithms, the average CPU running time versus SNR is shown in Fig.~\ref{Time}. As shown in Fig.~\ref{Time}, the NOMP algorithm requires a relatively long CPU running time because it updates the gradient matrix, Hessian matrix, and objective function at each iteration, resulting in high computational complexity. The GESBL algorithm also incurs a relatively high computational cost due to its iterative Bayesian posterior inference and repeated updates of the sparsity hyperparameters and off-grid parameters, which involve large-scale matrix operations in each iteration.} 
Furthermore, compared with NOMP, both the classical CP-ALS algorithm and the RecALS algorithm reduce the computational time by approximately one order of magnitude. 
This highlights the advantage of tensor decomposition-based methods in achieving low complexity when processing high-dimensional signals. 
Finally, since the proposed algorithm is based on Vandermonde-structured TT decomposition and involves no iterative process, the computational time of the proposed algorithm is further reduced by approximately one order of magnitude compared with the CP-ALS and RecALS algorithms. 
The low computational complexity of the proposed algorithm makes it highly suitable for deployment in high-mobility scenarios.

\begin{figure}[t]
	\centering
	\begin{minipage}[htbp]{1.0\linewidth}
	\centering
	\includegraphics[width=1.0\linewidth]{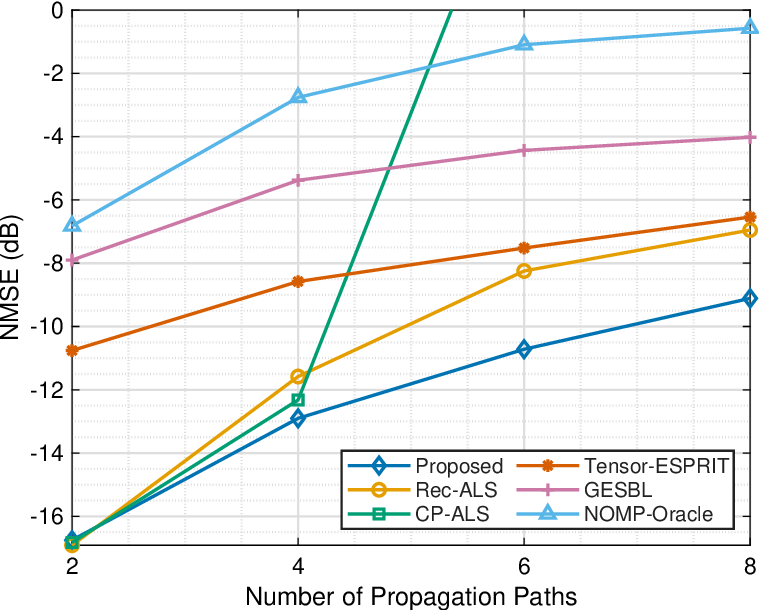}
	\caption{The NMSE performance versus the number of propagation paths, the SNR is 20 dB.}
	\label{NMSEvsPathNum}
	\end{minipage}%
	\vspace{-0.5cm}
	\hfill
\end{figure}

To further investigate the estimation performance of the algorithms under channels with different propagation scenarios, Fig.~\ref{NMSEvsPathNum} shows the NMSE performance versus the number of propagation paths, where the SNR is set to 20 dB. It can be observed from Fig.~\ref{NMSEvsPathNum} that the NMSE performance of the classical CP-ALS algorithm is mostly affected with the increase in the number of channel propagation paths. 
As the path number exceeds four, the NMSE performance of the CP-ALS deteriorates sharply.
{The main reason is the poor convergence behavior of the CP-ALS when processing the high-rank tensors \cite{alsconverge}.}
Furthermore, the performance degradation of the RecALS algorithm is less pronounced as the number of paths increases. This is attributed to the fact that, unlike the classical CP-ALS, the RecALS fully leverages the Vandermonde factor matrices within the tensor.
Finally, the proposed algorithm incurs less performance loss compared to the CP-ALS and the RecALS. In contrast to these two ALS-based algorithms, the proposed TT-based algorithm demonstrates superior robustness in processing high-rank tensor signals.

\begin{figure}[t]
	\centering
	\begin{minipage}[htbp]{1.0\linewidth}
	\centering
	\includegraphics[width=1.0\linewidth]{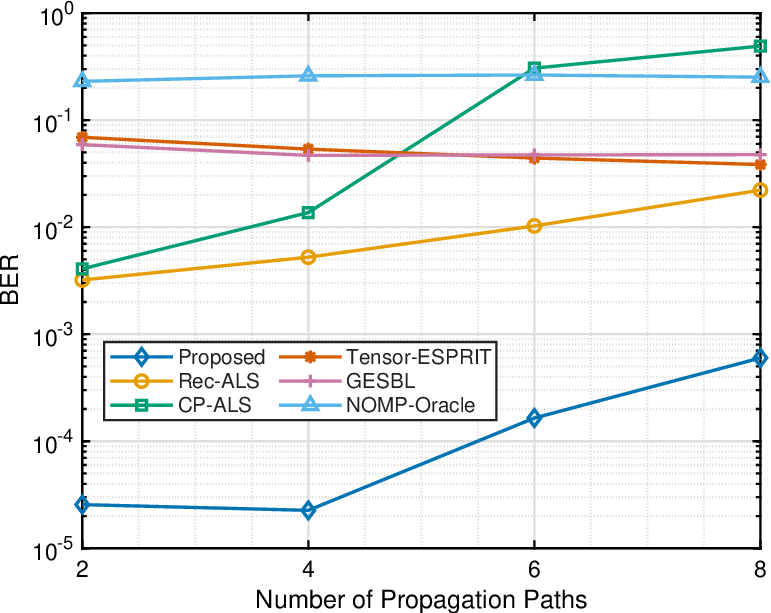}
	\caption{The BER performance versus the number of propagation paths, the SNR is 20 dB.}
	\label{BERvsPathNum}
	\end{minipage}%
	\vspace{-0.5cm}
	\hfill
\end{figure}

{
Fig.~\ref{BERvsPathNum} evaluates the BER performance versus the number of propagation paths. As illustrated in Fig.~\ref{BERvsPathNum}, as the number of channel paths increases from 2 to 8, the BERs of the proposed method and the ALS-based schemes deteriorate because the higher tensor rank increases the inter-path interference and reduces the separability of the spatial, temporal, and affine-frequency factors. Nevertheless, by exploiting the structured tensor representation of the AFDM channel, the proposed algorithm still achieves a lower BER than the benchmark schemes.
}

\begin{figure}[t]
	\centering
	\begin{minipage}[htbp]{1.0\linewidth}
	\centering
	\includegraphics[width=1.0\linewidth]{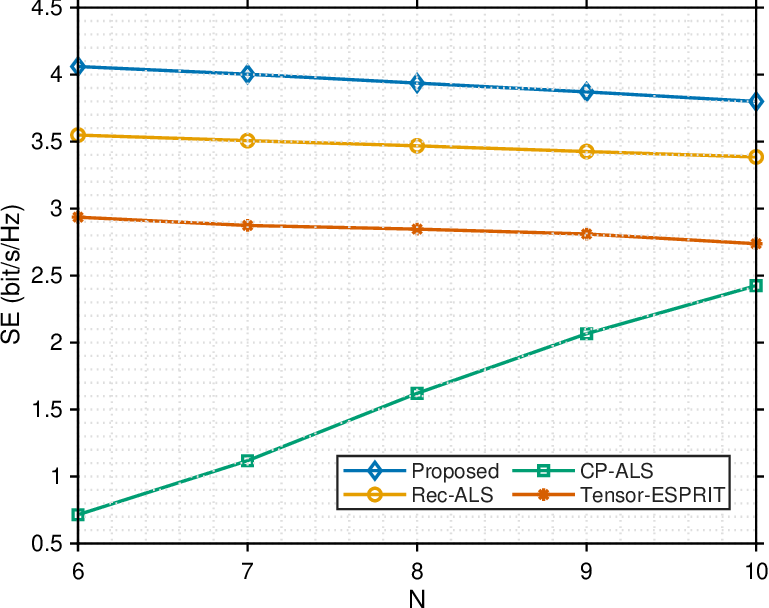}
	\caption{ The SE performance versus the number of pilot symbols, the SNR is 20 dB.}
	\label{SEvsN}
	\end{minipage}%
	\vspace{-0.5cm}
	\hfill
\end{figure}

\begin{figure}[t]
	\centering
	\begin{minipage}[htbp]{1.0\linewidth}
	\centering
	\includegraphics[width=1.0\linewidth]{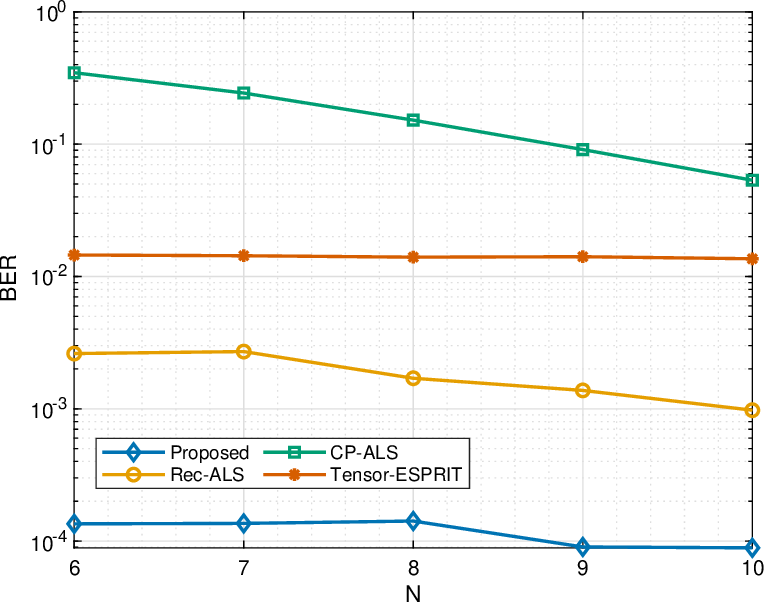}
	\caption{ The BER performance versus the number of pilot symbols, the SNR is 20 dB.}
	\label{BERvsN}
	\end{minipage}%
	\vspace{-0.5cm}
	\hfill
\end{figure}

{Fig.~\ref{SEvsN} and Fig.~\ref{BERvsN} illustrate the SE and BER performance, respectively, versus the number of pilot symbols $N$ for four tensor-based methods. As $N$ increases, the reduced correlation among the tensor factors improves the numerical conditioning of the CP decomposition, enabling CP-ALS to achieve substantial performance gains in terms of both a lower BER and a higher SE. In contrast, the other tensor-based methods achieve only modest BER reductions as $N$ increases, while their SE decreases because the limited improvement in channel estimation accuracy is insufficient to compensate for the additional pilot overhead introduced by a larger $N$.}

\subsection{Versus OTFS and OFDM}
{
To further investigate the performance comparison between AFDM, OTFS and OFDM in more realistic environments, we simulate their BER and SE performance under conditions of fractional delay and Doppler. The total number of subcarriers is set to $M=512$. The number of symbols in the OTFS frame is $N=13$. The number of pilots used for AFDM is set to $N_p=10$. The other parameters are the same as in Table \ref{simparameters}.
\begin{figure}[t]
	\centering
	\begin{minipage}[ht]{1.0\linewidth}
	\centering
	\includegraphics[width=1.0\linewidth]{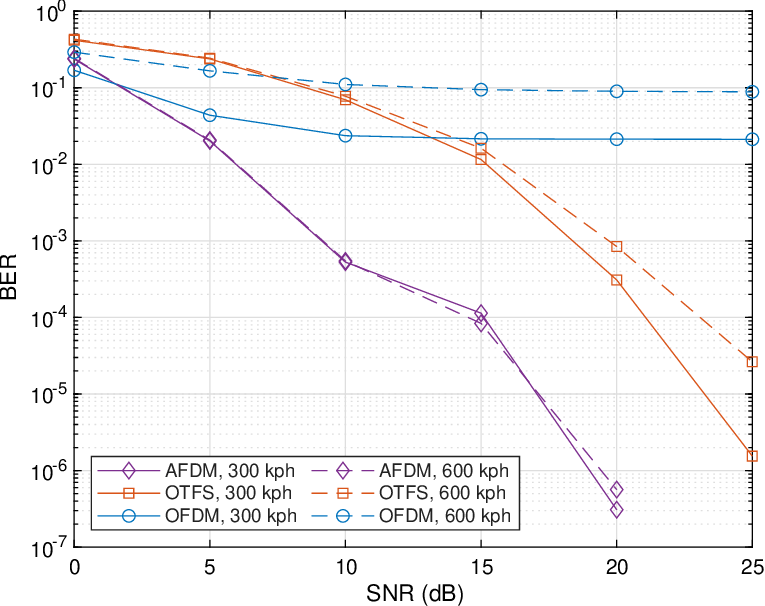}
	\caption{The BER performance versus SNR.}
	\label{AFDMvsOTFS_BER}
	\end{minipage}%
	\vspace{-0.5cm}
	\hfill
\end{figure}

\begin{figure}
	\centering
	\begin{minipage}[ht]{1.0\linewidth}
	\centering
	\includegraphics[width=1.0\linewidth]{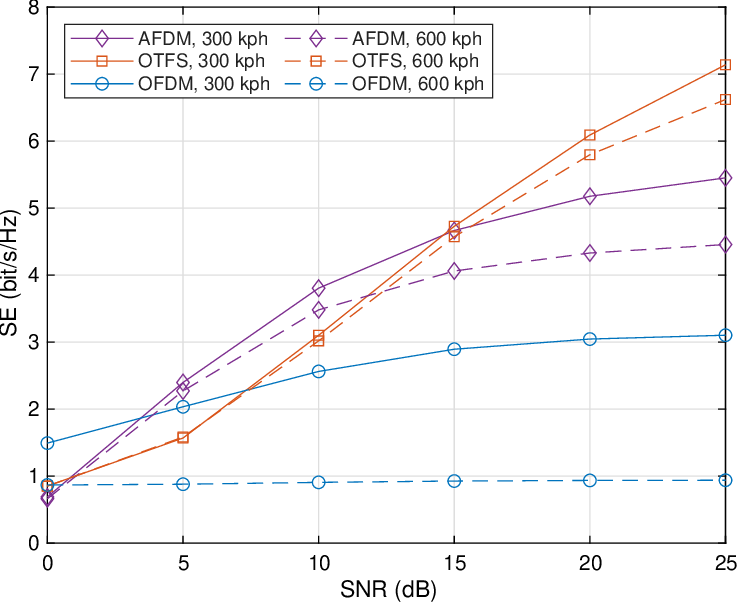}
	\caption{The SE performance versus SNR.}
	\label{AFDMvsOTFS_SE}
	\end{minipage}%
	\vspace{-0.5cm}
	\hfill
\end{figure}

Fig.~\ref{AFDMvsOTFS_BER} and Fig.~\ref{AFDMvsOTFS_SE} demonstrate the BER and SE performance, respectively, versus the SNR for AFDM, OTFS, and OFDM. As observed from these figures, AFDM achieves a lower BER than OTFS and OFDM over the considered SNR range, whereas OTFS attains a higher SE. This performance difference can be explained by the distinct modulation structures of AFDM and OTFS. Specifically, the symbol-wise temporal separation and DAFT-domain delay-Doppler separability of AFDM help localize interpath interference and channel-estimation errors, thereby improving the reliability of symbol detection. In contrast, OTFS spreads each delay-Doppler-domain symbol over the entire time-frequency frame through two-dimensional modulation, which enables more effective exploitation of time-frequency diversity and consequently yields a higher average achievable rate. However, this global coupling also causes residual fractional delay-Doppler interference and channel-estimation errors to propagate across multiple delay-Doppler symbols, resulting in a relatively higher uncoded BER under linear detection.
}

\section{Conclusion}\label{Conclusion}
In this paper, we investigated channel estimation for MIMO-AFDM systems and analyzed the input-output relationship for MIMO-AFDM systems while accounting for fractional delay and Doppler.
To tackle the estimation of fractional delay and Doppler, we proposed a T-AF domain embedded pilot structure.
By leveraging the magnitude of AF domain and the rotational invariance property of time domain, we proposed an efficient Vandermonde structure TT-based channel estimation algorithm. 
Compared to the existing parameter-estimation algorithms, the proposed algorithm significantly reduce the computational complexity.
To evaluate the MSE performance, we derived Ziv-Zakai bound for the 3D MIMO-AFDM signal model.
Simulation results showed that the proposed algorithm significantly outperforms the benchmarks at medium-to-high SNR.

\begin{appendices}
	
\section{Proof of Theorem 1}\label{proofoftheorem1}
To physically unveil the signal interaction in doubly dispersive channels, we adopt a derivation approach for the ambiguity function similar to \cite{raviteja} and \cite{ambiguity}. Although the derivation started in an integral form, the results take into account the process of discrete sampling and the application of the DAFT. As shown in the formula in the text, the derivation results of this paper are in complete agreement with the signal model in \cite{AFDM} under integer delay scenarios.

Based on (\ref{st})-(\ref{frdelaydoppler}), we obtain the following STAF-domain signal $Y\left[ n_{\mathrm{BS}},n,m \right] $ as given in (\ref{AF}), which is shown at the bottom of this page.
Then, the STAF domain signal $Y\left[ n_{\mathrm{BS}},n,m \right] $ can be obtained by discretely sampling the ambiguity function as (\ref{YSTAF}), which is shown at the bottom of the next page.
\begin{figure*}[hb]
	\hrulefill
	\begin{equation}\label{AF}
		\begin{aligned}
		&A_{g_{\mathrm{rx}},r_{n_{\mathrm{BS}}}}\left( t,f \right) \left| _{t=nT_{\mathrm{sym}},f=m\Delta f} \right. 
\\
&\quad=\int_{nT_{\mathrm{sym}}+T_{\mathrm{CPP}}}^{nT_{\mathrm{sym}}+T_{\mathrm{sym}}}{}e^{-j\left( \frac{d}{2b}\left( t^{\prime}-t-T_{\mathrm{CPP}} \right) ^2+2\pi \left( m\Delta f \right) \left( t^{\prime}-t-T_{\mathrm{CPP}} \right) +\frac{a}{2b}\left( m\Delta f \right) ^2 \right)}g_{\mathrm{rx}}^{\ast}\left( t^{\prime}-t \right) r_{n_{\mathrm{BS}}}\left( t^{\prime} \right) \mathrm{d}t^{\prime}
\\
&\quad=\sum_{n^{\prime}=0}^{N-1}{}\sum_{m^{\prime}=0}^{M-1}{}\int_{nT_{\mathrm{sym}}+T_{\mathrm{CPP}}}^{nT_{\mathrm{sym}}+T_{\mathrm{sym}}}{}e^{-j\left( \frac{d}{2b}\left( t^{\prime}-t-T_{\mathrm{CPP}} \right) ^2+2\pi \left( m\Delta f \right) \left( t^{\prime}-t-T_{\mathrm{CPP}} \right) +\frac{a}{2b}\left( m\Delta f \right) ^2 \right)}g_{\mathrm{rx}}^{\ast}\left( t^{\prime}-t \right) 
\\
&\quad\quad\times \left[ \iiint{}\left[ \mathbf{a}_{N_{\mathrm{BS}}}\left( \theta \right) \right] _{n_{\mathrm{BS}}}h\left( \tau ,\nu ,\theta \right) e^{j2\pi \nu \left( t^{\prime}-\tau \right)}x_{n^{\prime},m^{\prime}}\phi _{m^{\prime}}\left( t^{\prime}-n^{\prime}T_{\mathrm{sym}}-\tau \right) g_{\mathrm{tx}}\left( t^{\prime}-n^{\prime}T_{\mathrm{sym}}-\tau \right) \mathrm{d}\tau \mathrm{d}\nu \mathrm{d}\theta \right] \mathrm{d}t^{\prime}
\\
&\quad=\sum_{m^{\prime}=0}^{M-1}{}x_{n,m^{\prime}}\iiint{}\left[ \mathbf{a}_{N_{\mathrm{BS}}}\left( \theta \right) \right] _{n_{\mathrm{BS}}}h\left( \tau ,\nu ,\theta \right) 
\\
&\quad\quad\times \int_{nT_{\mathrm{sym}}+T_{\mathrm{CPP}}}^{nT_{\mathrm{sym}}+T_{\mathrm{sym}}}{}e^{j2\pi \nu \left( t^{\prime}-\tau \right)}e^{-j\left( \frac{d}{2b}\left( t^{\prime}-t-T_{\mathrm{CPP}} \right) ^2+2\pi \left( m\Delta f \right) \left( t^{\prime}-t-T_{\mathrm{CPP}} \right) +\frac{a}{2b}\left( m\Delta f \right) ^2 \right)}\phi _{m^{\prime}}\left( t^{\prime}-nT_{\mathrm{sym}}-\tau \right) \mathrm{d}t^{\prime}
\\
&\quad=\sum_{m^{\prime}=0}^{M-1}{}x_{n,m^{\prime}}\sum_{i=1}^{P}{}\alpha _i\left[ \mathbf{a}_{N_{\mathrm{BS}}}\left( \theta \right) \right] _{n_{\mathrm{BS}}}
\\
&\quad\quad\times \int_{nT_{\mathrm{sym}}+T_{\mathrm{CPP}}}^{nT_{\mathrm{sym}}+T_{\mathrm{sym}}}{}e^{j2\pi \nu _i\left( t^{\prime}-\tau _i \right)}e^{-j\left( \frac{d}{2b}\left( t^{\prime}-t-T_{\mathrm{CPP}} \right) ^2+2\pi \left( m\Delta f \right) \left( t^{\prime}-t-T_{\mathrm{CPP}} \right) +\frac{a}{2b}\left( m\Delta f \right) ^2 \right)}\phi _{m^{\prime}}\left( t^{\prime}-nT_{\mathrm{sym}}-\tau _i \right) \mathrm{d}t^{\prime}
		\end{aligned}
	\end{equation}

	\hrulefill
	\begin{equation}\label{YSTAF}
		\begin{aligned}
			Y\left[ n_{\mathrm{BS}},n,m \right] &=\sum_{m^{\prime}=0}^{M-1}{}x_{n,m^{\prime}}\sum_{i=1}^{P}{}\tilde{\alpha}_i\left[ \mathbf{a}_{N_{\mathrm{BS}}}\left( \theta \right) \right] _{n_{\mathrm{BS}}}\frac{1}{M}\sum_{p=M_{\mathrm{CPP}}}^{M+M_{\mathrm{CPP}}-1}{}e^{j2\pi \nu _i\frac{T}{M}p}e^{-j2\pi c_1\left( p-M_{\mathrm{CPP}} \right) ^2}e^{-j2\pi \frac{m\left( p-M_{\mathrm{CPP}} \right)}{M}}e^{-j2\pi c_2m^2}
\\
&\quad\times e^{j2\pi c_1\left( p-l_i-\iota _i-M_{\mathrm{CPP}} \right) ^2}e^{j2\pi \frac{m^{\prime}\left( p-l_i-\iota _i-M_{\mathrm{CPP}} \right)}{M}}e^{j2\pi c_2{m^{\prime}}^2}
\\
&=\sum_{m^{\prime}=0}^{M-1}{}x_{n,m^{\prime}}\sum_{i=1}^{P}{}\tilde{\alpha}_i\left[ \mathbf{a}_{N_{\mathrm{BS}}}\left( \theta \right) \right] _{n_{\mathrm{BS}}}e^{j2\pi \nu _inT_{\mathrm{sym}}}H_i\left[ m,m^{\prime} \right]
		\end{aligned}
	\end{equation}
\end{figure*}

\section{Proof of Theorem 2}\label{proofoftheorem2}
Using equation $\frac{\partial \mathbf{X}\left( s \right) ^{-1}}{\partial s}=-\mathbf{X}\left( s \right) ^{-1}\frac{\partial \mathbf{X}\left( s \right)}{\partial s}\mathbf{X}\left( s \right) ^{-1}$, the second-order partial derivative of $\mu \left( s;\mathbf{{\bm \delta} } \right) $ w.r.t. $s$ is given by
	\begin{align}
		\frac{\partial ^2\mu \left( s;\mathbf{{\bm \delta} } \right)}{\partial s^2}&=\mathrm{Tr}\left( \left( \left[ s\mathbf{R}_{\mathbf{y}|\mathbf{{\bm \varphi} }+\mathbf{{\bm \delta} }}^{-1}+\left( 1-s \right) \mathbf{R}_{\mathbf{y}|\mathbf{{\bm \varphi} }}^{-1} \right] ^{-1}\right.\right.\nonumber
		\\
		&\quad\left.\left.\times\left( \mathbf{R}_{\mathbf{y}|\mathbf{{\bm \varphi} }+\mathbf{{\bm \delta} }}^{-1}-\mathbf{R}_{\mathbf{y}|\mathbf{{\bm \varphi} }}^{-1} \right) \right) ^2 \right).
	\end{align} 

We first examine the values of $\mu \left( s;\mathbf{{\bm \delta} } \right) $ and $\frac{\partial ^2\mu \left( s;\mathbf{{\bm \delta} } \right)}{\partial s^2}$ in the asymptotic regime, i.e., in the vicinity of $\mathbf{{\bm \delta} }=\mathbf{0}_{3P}$. Performing the first-order Taylor expansion of $\mathbf{R}_{\mathbf{y}|\mathbf{{\bm \varphi} }+\mathbf{{\bm \delta} }}$ as
\begin{equation}
	\mathbf{R}_{\mathbf{y}|\mathbf{{\bm \varphi} }+\mathbf{{\bm \delta} }}\approx \mathbf{R}_{\mathbf{y}|\mathbf{{\bm \varphi} }}+\underset{\Delta \mathbf{R}}{\underbrace{\sum_i{}{\bm \delta} _i\frac{\partial \mathbf{R}_{\mathbf{y}|\mathbf{{\bm \varphi} }}}{\partial p_i}\;\mid_{\mathbf{{\bm \varphi} }=\bm{p}}^{}}}.
\end{equation}
Then, we can expand the function $\mu \left( s;\mathbf{{\bm \delta} } \right)$ by applying the approximation $\left( \mathbf{I}+\mathbf{X} \right) ^{-1}\approx \mathbf{I}-\mathbf{X}$ and $\ln \left| \mathbf{I}+\mathbf{X} \right|\approx \mathrm{Tr}\left( \mathbf{X} \right) -\frac{1}{2}\mathrm{Tr}\left( \mathbf{X}^2 \right) $ as
\begin{align}
		\mu \left( s;\mathbf{{\bm \delta} } \right) &\approx -\ln \left| s\left( \mathbf{I}+\mathbf{R}_{\mathbf{y}|\mathbf{{\bm \varphi} }}^{-1}\Delta \mathbf{R} \right) ^{-1}+\left( 1-s \right) \mathbf{I} \right|\nonumber
\\
&\quad-s\ln \left| \mathbf{I}+\mathbf{R}_{\mathbf{y}|\mathbf{{\bm \varphi} }}^{-1}\Delta \mathbf{R} \right|\nonumber
\\
&\approx s\mathrm{Tr}\left( \mathbf{R}_{\mathbf{y}|\mathbf{{\bm \varphi} }}^{-1}\Delta \mathbf{R} \right) +\left( \frac{s^2}{2}-s \right) \mathrm{Tr}\left( \left( \mathbf{R}_{\mathbf{y}|\mathbf{{\bm \varphi} }}^{-1}\Delta \mathbf{R} \right) ^2 \right) \nonumber
\\
&\quad-s\mathrm{Tr}\left( \mathbf{R}_{\mathbf{y}|\mathbf{{\bm \varphi} }}^{-1}\Delta \mathbf{R} \right) +\frac{s}{2}\mathrm{Tr}\left( \left( \mathbf{R}_{\mathbf{y}|\mathbf{{\bm \varphi} }}^{-1}\Delta \mathbf{R} \right) ^2 \right) \nonumber
\\
&=\frac{s^2-s}{2}\mathrm{Tr}\left( \left( \mathbf{R}_{\mathbf{y}|\mathbf{{\bm \varphi} }}^{-1}\Delta \mathbf{R} \right) ^2 \right) =\frac{s^2-s}{2}\mathbf{{\bm \delta} }^{\mathrm{T}}\mathbf{J{\bm \delta} }.\label{ua}
\end{align}
Hence, we have $\left.\mu \left( s;\mathbf{{\bm \delta} } \right)\right|_{s=\frac{1}{2}}^{} =-\frac{1}{8}\mathbf{{\bm \delta} }^{\mathrm{T}}\mathbf{J{\bm \delta} }$.
Similarly, $\frac{\partial ^2\mu \left( s;\mathbf{{\bm \delta} } \right)}{\partial s^2}$ can be rewritten as
\begin{align}
		\frac{\partial ^2\mu \left( s;\mathbf{{\bm \delta} } \right)}{\partial s^2}&=\mathrm{Tr}\left( \left( \left[ s\left( \mathbf{R}_{\mathbf{y}|\mathbf{{\bm \varphi} }}^{}+\Delta \mathbf{R} \right) ^{-1}+\left( 1-s \right) \mathbf{R}_{\mathbf{y}|\mathbf{{\bm \varphi} }}^{-1} \right] ^{-1}\right.\right.\nonumber
		\\
		&\quad\times\left.\left.\left( \left( \mathbf{R}_{\mathbf{y}|\mathbf{{\bm \varphi} }}^{}+\Delta \mathbf{R} \right) ^{-1}-\mathbf{R}_{\mathbf{y}|\mathbf{{\bm \varphi} }}^{-1} \right) \right) ^2 \right) \nonumber
\\
&\approx \mathrm{Tr}\left( \left( \left[ \mathbf{R}_{\mathbf{y}|\mathbf{{\bm \varphi} }}^{}+\Delta \mathbf{R} \right] \mathbf{R}_{\mathbf{y}|\mathbf{{\bm \varphi} }}^{-1}\Delta \mathbf{RR}_{\mathbf{y}|\mathbf{{\bm \varphi} }}^{-1} \right) ^2 \right) \nonumber
\\
&\approx \mathrm{Tr}\left( \left( \Delta \mathbf{RR}_{\mathbf{y}|\mathbf{{\bm \varphi} }}^{-1} \right) ^2 \right) \nonumber
\\
&=\mathbf{{\bm \delta} }^{\mathrm{T}}\mathbf{J{\bm \delta} }.
\end{align}
Next, we analyze the behavior of the function $\mu \left( s;\mathbf{{\bm \delta} } \right) $ and $\frac{\partial ^2\mu \left( s;\mathbf{{\bm \delta} } \right)}{\partial s^2}$ in the non-asymptotic regime. When $\mathbf{{\bm \delta} }$ grows large, the array steering vector $\mathbf{a}_{N_{\mathrm{BS}}}\left( \theta _i \right) $ and $\mathbf{a}_{N_{\mathrm{BS}}}\left( \theta _i+{\bm \delta} _{\theta _i} \right) $ exhibit asymptotic orthogonality \cite{zzb2d}. A similar trend is also observed for $\mathbf{b}_N\left( \nu _i \right) $ and $\mathbf{b}_N\left( \nu _i+{\bm \delta} _{\nu _i} \right) $. Furthermore, as $\mathbf{{\bm \delta} }_{\eta _i}$ increases, the main-lobe positions of $\mathbf{c}_{m_{\mathrm{pilot}}}\left( \nu _i,\eta _i \right) $ and $\mathbf{c}_{m_{\mathrm{pilot}}}\left( \nu _i,\eta _i+\mathbf{{\bm \delta} }_{\eta _i} \right) $ become misaligned, and the correlation $\left| \mathbf{c}_{m_{\mathrm{pilot}}}\left( \nu _i,\eta _i \right) ^{\mathrm{H}}\mathbf{c}_{m_{\mathrm{pilot}}}\left( \nu _i,\eta _i+\mathbf{{\bm \delta} }_{\eta _i} \right) \right|$ gradually approaches zero. Therefore, it can be inferred that, as $\mathbf{{\bm \delta} }$ increases asymptotically, the corresponding subspaces become asymptotically orthogonal. The two covariance matrices contribute equally to the eigenvalues in the signal subspace, both being $\frac{1}{\left( \sigma _{n}^{2}+N_{\mathrm{BS}}N\sigma _{i,\mathrm{PL}}^{2} \right)}$, while the eigenvalue contributions in the remaining noise subspace are $\frac{1}{\sigma _{n}^{2}}$. Thus, we have
\begin{equation}
	\left.\mu \left( s;\mathbf{{\bm \delta} } \right) \right|_{s=\frac{1}{2}}^{}\approx \sum_{i=0}^{P-1}{}\ln \left[ \frac{4\left( 1+N_{\mathrm{BS}}N\rho _i \right)}{\left( 2+N_{\mathrm{BS}}N\rho _i \right) ^2} \right],
\end{equation}
and 
\begin{equation}\label{pu2na}
	\left.\frac{\partial ^2\mu \left( s;\mathbf{{\bm \delta} } \right)}{\partial s^2}\right|_{s=\frac{1}{2}}^{}=2\sum_{i=0}^{P-1}{}\left( \frac{2N_{\mathrm{BS}}N\rho _i}{2+N_{\mathrm{BS}}N\rho _i} \right) ^2.
\end{equation}
Substituting (\ref{ua})-(\ref{pu2na}) into (\ref{Pdelta}), we obtain (\ref{Pdeltaexpression}), which completes the proof.

\end{appendices}

\bibliographystyle{IEEEtran}
% %argument is your BibTeX string definitions and bibliography database(s)
\bibliography{myre}

% that's all folks

\end{document}